\documentclass{iopjournal}

\usepackage{amsmath}
\usepackage{amssymb}
\usepackage{amsthm}
\usepackage{mathtools}
\usepackage{empheq}
\usepackage{bm}
\usepackage{upgreek}
\usepackage{enumitem}
\usepackage{hyperref}
\usepackage[capitalize]{cleveref}

\newtheorem{theorem}{Theorem}[section]
\newtheorem{proposition}[theorem]{Proposition}
\newtheorem{lemma}[theorem]{Lemma}
\newtheorem{corollary}[theorem]{Corollary}
\theoremstyle{definition}
\newtheorem{definition}[theorem]{Definition}
\theoremstyle{remark}
\newtheorem{remark}[theorem]{Remark}

\newcommand{\thz}{\theta_{0}}
\newcommand{\mf}{m_{f}}
\newcommand{\ms}{m_{s}}
\newcommand{\lam}{\lambda}
\newcommand{\Bx}{\Box}

\newcommand{\CC}{\mathbb{C}}
\newcommand{\RR}{\mathbb{R}}

\newcommand{\slC}{\mathfrak{sl}(2,\mathbb{C})}
\newcommand{\slR}{\mathfrak{sl}(2,\mathbb{R})}
\newcommand{\aslhat}{\widehat{\mathfrak{sl}}(2)^{(1)}}

\newcommand{\genH}{\mathsf{H}}
\newcommand{\genEp}{\mathsf{E}_{+}}
\newcommand{\genEpm}{\mathsf{E}_{\pm}}
\newcommand{\genEm}{\mathsf{E}_{-}}

\newcommand{\Aplus}{A_{+}}
\newcommand{\Aminus}{A_{-}}

\newcommand{\dpp}{\partial_{+}}
\newcommand{\dmm}{\partial_{-}}

\newcommand{\bp}{\bar{\psi}}

\newcommand{\im}{\mathrm{i}}
\newcommand{\ee}{\mathrm{e}}

\newcommand{\Tr}{\operatorname{tr}}
\newcommand{\Diag}{\operatorname{diag}}
\newcommand{\Ad}{\operatorname{Ad}}

\begin{document}

\title{Integrable deformations of the
       Dirac--sinh-Gordon system}

\author{Laith H. Haddad\orcid{0009-0001-4586-1643}}

\affil{$^1$Department of Physics, Colorado School of Mines, Golden, Colorado, USA}

\affil{$^*$Author to whom any correspondence should be addressed.}

\email{lhaddad@mines.edu}


\begin{abstract}
We construct a two-dimensional family of integrable coupled
Dirac--scalar field theories in $1+1$ dimensions, parameterized
by $(\thz,\alpha)\in[0,\pi/2]^2$, whose Lax connection takes
values in $\slC$ throughout.  The family arises as the orbit of
the Dirac--sinh-Gordon system under the $U(1)\times U(1)$ maximal
torus of $\mathrm{Inn}(\slC)\cong PSL(2,\CC)$: the $\thz$-$U(1)$
rotates the constant phase of the Dirac mass; the $\alpha$-$U(1)$
rotates its field-dependent phase via $\beta\to|\beta|\ee^{\im\alpha}$.
Integrability throughout the parameter space follows from a single
principle: any automorphism of the Lax algebra preserves the
zero-curvature condition, since the condition depends only on the
Lie bracket of $\slC$.  The parameter square is the positive Weyl
chamber of $T^2\subset SU(2)/\{\pm I\}$, with the four corner
theories corresponding to the four real forms of $\slC$ accessible
within the Leznov--Saveliev grading.

The two axes have opposite non-Hermitian characters.  The
$\thz$-axis is a \emph{decoupling parameter}: the effective
coupling $g_{\rm eff} = g\cos\thz$ is forced by the zero-curvature
condition, $\sin\thz$ is the $\mathcal{PT}$-breaking order
parameter measured by the anomalous divergence of the fermion
number current, and the Dirac mass phase rotates from real to
purely imaginary.  The $\alpha$-axis preserves fermion number
exactly for all $\alpha$, provides a fully-coupled integrable
path from the Dirac--sinh-Gordon to the Dirac--sine-Gordon system
with $\phi$ real throughout, and is identified as the real-form
transition $\slR\to\mathfrak{su}(1,1)$ underlying the
Coleman--Mandelstam bosonization between the massive Thirring
model and the sine-Gordon system.

For the $\thz$-axis we give complete proofs: an explicit
$\slC$-valued zero-curvature representation verified grade by
grade; physical non-triviality (distinct $\thz$ define genuinely
inequivalent theories, established via a gauge analysis that
leaves the fermion bilinear invariant); an anomalous continuity
equation with $\mathcal{PT}$-breaking order parameter $\sin\thz$;
and an infinite tower of conserved charges related to the
sinh-Gordon hierarchy by $\thz$-dependent phase factors, yielding
a doubled conservation law structure.  The identification
$\eta = \tan(\thz/2)$ places the construction within the
Yang--Baxter deformation framework for the $SL(2)$ principal
chiral model.
\end{abstract}

\section{Introduction}
\label{sec:intro}

The sinh-Gordon equation in $1+1$ spacetime dimensions,
\begin{equation}
  \label{eq:sinhG}
  \Bx\phi + \frac{\ms^2}{\beta}\sinh(\beta\phi) = 0,
  \qquad \Bx\equiv \partial_t^2-\partial_x^2,
\end{equation}
is one of the canonical integrable nonlinear field theories.  It
belongs to the affine Toda field theory hierarchy associated with
the untwisted affine algebra $\aslhat$, possesses an explicit Lax
pair in the Leznov--Saveliev
form~\cite{LeznovSaveliev1979,Mikhailov1981}, an infinite tower of
polynomial conserved charges~\cite{BulloughDodd1977}, and admits
exact multi-soliton solutions via the inverse scattering
transform~\cite{AblowitzSegur1981,ZakharovShabat1972}.
Coupling~\eqref{eq:sinhG} to a Dirac spinor $\psi$ in $1+1$
dimensions yields the Dirac--sinh-Gordon system
\begin{align}
  \label{eq:DSG_scalar}
  \Bx\phi + \frac{\ms^2}{\beta}\sinh(\beta\phi) &= g\,\bp\psi,\\
  \label{eq:DSG_Dirac}
  \left(\im\gamma^\mu\partial_\mu - \mf\ee^{\beta\phi}\right)\psi &= 0,
\end{align}
which is integrable and has been studied in the context of soliton
theory and affine Toda field theory with
matter~\cite{Getmanov1977,ZakharovMikhailov1978,
OliveTurokUnderwood1993a,OliveTurokUnderwood1993b,
BabelonBonora1990,FerreiraMiramontes1997}.
A second closely related integrable system is the Dirac--sine-Gordon
system,
\begin{align}
  \label{eq:DsG_scalar}
  \Bx\phi + \frac{\ms^2}{\beta}\sin(\beta\phi) &= g\,\bp\psi,\\
  \label{eq:DsG_Dirac}
  \left(\im\gamma^\mu\partial_\mu
    - \mf\ee^{\im\beta\phi}\right)\psi &= 0,
\end{align}
which is equivalent, via the Coleman--Mandelstam
bosonization~\cite{Coleman1975,Mandelstam1975}, to the massive
Thirring model~\cite{Thirring1958,KlaiberBarut1968}.  The quantum
integrability of the massive Thirring model has been established by
several methods, including exact $S$-matrix
approaches~\cite{BergKarowski1979,ZamolodchikovZamolodchikov1979}
and the quantum inverse scattering method~\cite{IzerginKorepin1981}.
The two systems differ in the coupling appearing in the Dirac
equation: a real exponential $\ee^{\beta\phi}$
in~\eqref{eq:DSG_Dirac} versus a purely imaginary exponential
$\ee^{\im\beta\phi}$ in~\eqref{eq:DsG_Dirac}.  Systems with
complex or non-Hermitian mass terms of this type also arise in the
context of $\mathcal{PT}$-symmetric field
theories~\cite{BenderBoettcher1998,BenderBrody2002,ZnojilZnojil2008},
where the reality of the spectrum is maintained not by Hermiticity
but by the combined $\mathcal{PT}$ symmetry of the Hamiltonian.
As shown in Section~\ref{sec:bilinear}, the two axes of the
parameter space $(\thz,\alpha)$ introduced in this paper have
opposite non-Hermitian characters: the $\thz$-axis breaks fermion
number conservation and is the $\mathcal{PT}$-breaking direction,
while the $\alpha$-axis preserves fermion number exactly for all
$\alpha$ and is the unitarity-preserving direction.

A natural question is: \emph{what is the full space of integrable
deformations of the Dirac--sinh-Gordon system accessible within
the same Lax structure?}  The main result of this paper is that
this space is two-dimensional, parameterized by
$(\thz,\alpha)\in[0,\pi/2]^2$, where $\thz = \arg(e^{i\thz})$
rotates the constant phase of the Dirac mass and
$\alpha = \arg\beta$ rotates its field-dependent part.  Both
parameters enter the Lax connection as automorphisms of $\slC$
that individually preserve the zero-curvature condition, and
they are independent.

The primary focus of this paper is the \emph{$\thz$-axis}
($\alpha=0$).  We introduce the \emph{$\thz$-deformed
Dirac--sinh-Gordon system}:
\begin{align}
  \label{eq:theta_scalar}
  \Bx\phi + \frac{\ms^2}{\beta}\sinh(\beta\phi)
    &= g\cos\thz\;\bp\psi,\\
  \label{eq:theta_Dirac}
  \left(\im\gamma^\mu\partial_\mu
    - \mf\ee^{\im\thz}\ee^{\beta\phi}\right)\psi &= 0,
\end{align}
where $\thz\in[0,\pi/2]$ is a \emph{decoupling parameter}: the
effective scalar--fermion coupling is $g_{\rm eff}(\thz) =
g\cos\thz$, which decreases monotonically from the full value $g$
at $\thz=0$ to zero at $\thz=\pi/2$.  We prove that the system is
integrable for every $\thz$ via an explicit zero-curvature
representation, and establish that $\sin\thz$ is the
$\mathcal{PT}$-breaking order parameter of the family.  The
complementary \emph{$\alpha$-axis} ($\thz=0$, $\beta =
|\beta|\ee^{\im\alpha}$) provides a fully-coupled integrable
path from the Dirac--sinh-Gordon to the Dirac--sine-Gordon system
with backreaction coefficient $g$ unchanged and $\phi$ real
throughout; this axis preserves fermion number exactly and is the
unitarity-preserving direction of the parameter space.  Both axes
and their physical and algebraic interpretations are developed
in Section~\ref{sec:zc} and Section~\ref{sec:algebra}.

The $\cos\thz$ suppression of the backreaction is not put in by
hand: it is forced by the zero-curvature condition, which requires
the backreaction coefficient to equal the real part of the complex
Dirac mass $M = \mf\ee^{\im\thz}\ee^{\beta\phi}$, namely
$\operatorname{Re}(M + M^*) = 2\mf\cos\thz\,\ee^{\beta\phi}$
(established in Section~\ref{sec:ZC_grade0} and
Remark~\ref{rmk:cosine_forced}).  The coupling switches off and
the Dirac mass rotates simultaneously because they are two aspects
of the same underlying operation: the deformation is the orbit of
the Dirac--sinh-Gordon system under a $U(1)$ action that
simultaneously rotates the scalar field's real section in the
complex plane and generates the $U(1)$ automorphism of $\slC$
that deforms the reality condition on the Lax connection
(Section~\ref{sec:U1}).  As the coupling decreases to zero, the
Dirac mass rotates from real to purely imaginary, arriving at
the Dirac--sine-Gordon structure at complete decoupling.

At $\thz=0$ the system reduces exactly to the fully coupled
Dirac--sinh-Gordon model~\eqref{eq:DSG_scalar}--\eqref{eq:DSG_Dirac}.
At $\thz=\pi/2$ the backreaction term $g\cos\thz\,\bp\psi$ is
indeterminate in the original field $\psi$ — the factor
$\cos\thz\to 0$ while $\bp\psi$ diverges simultaneously — and
its resolution depends on the field description one adopts.  In
the original field, the backreaction switches off and the system
becomes a free sine-Gordon scalar plus a Dirac fermion in the
sine-Gordon background, a well-defined and exactly integrable
endpoint.  In terms of the rescaled field
$\tilde\psi = \sqrt{\cos\thz}\,\psi$, the backreaction is
$g\,\bar{\tilde\psi}\tilde\psi$ for all $\thz$, and the
$\thz\to\pi/2$ limit yields the fully coupled
Dirac--sine-Gordon system after the analytic continuation
$\phi\to-\im\varphi$.  Both descriptions are exact; the
original-field description makes the decoupling transparent,
while the rescaled-field description makes the endpoint
identification precise.  The connection to the massive Thirring
model via Coleman--Mandelstam bosonization~\cite{Coleman1975,
Mandelstam1975} holds at the sine-Gordon endpoint in either
description.  The full endpoint analysis, including the role
of the analytic continuation and the singular fermion rescaling,
is carried out in Section~\ref{sec:setup}.

The integrability of the deformation is established by constructing
an explicit $\slC$-valued zero-curvature representation valid for
all $\thz\in[0,\pi/2]$.  The Lax pair is written in a
Leznov--Saveliev form~\cite{LeznovSaveliev1979} adapted to the
affine $\aslhat$ Toda structure, with the deformation entering
through constant half-angle phase factors in the off-diagonal
components.  The key mechanism is that in every commutator appearing
in the zero-curvature condition, the phase factors enter as
conjugate pairs $\ee^{+\im\thz/2}\cdot\ee^{-\im\thz/2} = 1$, so
the curvature is $\thz$-independent at the Lax level and the field
equations including the precise factor of $\cos\thz$ are recovered
entirely from the grade-$0$ equation of motion matching.
At the same time, the deformation is not a trivial rewriting of the
Dirac--sinh-Gordon model: although the deformed Lax connection is
related to the $\thz=0$ connection by a constant complex gauge
transformation $h_{\thz}=\Diag(\ee^{-\im\thz/4},\ee^{+\im\thz/4})$,
this transformation leaves the fermion bilinear $\bp\psi$ invariant
while changing the Dirac coupling constant, producing inequivalent
field theories for distinct values of $\thz$.  The ratio of the
Dirac coupling phase to the scalar backreaction coefficient is
therefore a genuine physical observable that labels the degree of
decoupling.  The Yang--Baxter parameter $\eta = \tan(\thz/2)$
is a stereographic coordinate on this decoupling orbit, running
from $\eta=0$ at full coupling to $\eta\to\infty$ at complete
decoupling, and providing a direct link between the physical
coupling and the non-commutativity of the target space geometry
in the $\eta$-deformed sigma model.

The paper is organised as follows.
Section~\ref{sec:setup} establishes conventions and verifies the
endpoint systems.  Section~\ref{sec:lax} introduces the Lax pair.
Section~\ref{sec:zc} carries out the complete zero-curvature
computation including the detailed derivation of $\cos\thz$, and
establishes the two-dimensional integrable parameter space
(Corollary~\ref{cor:beta} and Proposition~\ref{prop:2d}).
Section~\ref{sec:gauge} analyses gauge-equivalence and proves
physical non-triviality.  Section~\ref{sec:bilinear} derives the
anomalous continuity equation, identifies $\sin\thz$ as the
$\mathcal{PT}$-breaking order parameter, and proves fermion number
conservation along the $\alpha$-axis.
Section~\ref{sec:charges} constructs three conserved charges
explicitly via AKNS recursion.  Section~\ref{sec:algebra} discusses
the algebraic interpretation, the Yang--Baxter identification, and
the automorphism principle that unifies integrability and
$\mathcal{PT}$ symmetry across the full parameter space.
Section~\ref{sec:conclusion} summarises results and lists open
problems.

\section{Setup and endpoint systems}
\label{sec:setup}

\subsection{Conventions}

We work in $1+1$-dimensional Minkowski space with metric
$\eta = \Diag(+,-)$.  Light-cone coordinates are
$x^{\pm} = t \pm x$ with $\partial_{\pm} = \frac{1}{2}(\partial_t\pm\partial_x)$,
giving $\partial_t = \partial_+ + \partial_-$,
$\partial_x = \partial_+ - \partial_-$, and
$\Box = 4\partial_+\partial_-$. For the Dirac algebra we use
$\gamma^0 = \begin{psmallmatrix}1&0\\0&-1\end{psmallmatrix}$,
$\gamma^1 = \begin{psmallmatrix}0&1\\-1&0\end{psmallmatrix}$,
satisfying $\{\gamma^\mu,\gamma^\nu\} = 2\eta^{\mu\nu}$.  We write
$\psi = \begin{psmallmatrix}\psi_+\\\psi_-\end{psmallmatrix}$ and
$\bp = \psi^\dagger\gamma^0$.  The fermion bilinear is therefore
$\bp\psi = |\psi_+|^2 - |\psi_-|^2$. The $\slC$ generators in the fundamental representation are
\begin{equation}
  \label{eq:sl2}
  \genH = \begin{pmatrix}1&0\\0&-1\end{pmatrix},\quad
  \genEp = \begin{pmatrix}0&1\\0&0\end{pmatrix},\quad
  \genEm = \begin{pmatrix}0&0\\1&0\end{pmatrix},
\end{equation}
with commutation relations
$[\genH,\genEp]=2\genEp$, $[\genH,\genEm]=-2\genEm$,
$[\genEp,\genEm]=\genH$.

\subsection{Component form of the Dirac equation}

With the conventions above, the Dirac
equation~\eqref{eq:theta_Dirac} expands in light-cone coordinates
as
\begin{align}
  \label{eq:Dpm_lc}
  \im\dpp\psi_- &= \tfrac{1}{2}\mf\ee^{\im\thz}\ee^{\beta\phi}\psi_+,\\
  \label{eq:Dmm_lc}
  \im\dmm\psi_+ &= \tfrac{1}{2}\mf\ee^{\im\thz}\ee^{\beta\phi}\psi_-.
\end{align}
These equations have a transparent structure: each light-cone
derivative of one spinor component is sourced by the other component
through the complex Yukawa coupling $\mf\ee^{\im\thz}\ee^{\beta\phi}$.
The two equations are related by the exchange $\psi_+\leftrightarrow\psi_-$
combined with $\partial_+\leftrightarrow\partial_-$, reflecting the
underlying $\mathbb{Z}_2$ symmetry of the light-cone decomposition. The complex phase $\ee^{\im\thz}$ rotates the mass term in the
complex plane, from the purely real coupling at $\thz=0$ to the
purely imaginary coupling at $\thz=\pi/2$.  As $\thz$ increases,
this rotation is accompanied by the suppression of the backreaction
coefficient $\cos\thz$: the Dirac mass rotation and the decoupling
of the scalar sector are two aspects of the same underlying $U(1)$
action (Section~\ref{sec:U1}).  It is
precisely this half-angle structure in the Dirac equation that will
reappear, as the half-angle phases $\ee^{\pm\im\thz/2}$, in the
off-diagonal entries of the Lax pair in the next section.

\subsection{Endpoint systems}

We now verify the two endpoint systems of the deformation family.
The $\thz=0$ endpoint is immediate; the $\thz=\pi/2$ endpoint
requires two additional steps beyond setting the parameter value —
a fermion rescaling and an analytic continuation of the scalar
field — each of which we discuss carefully.
\begin{itemize}[leftmargin=2em]
\item \textit{Endpoint $\thz=0$ (Dirac--sinh-Gordon).}
  At $\thz=0$ the coupling $\mf\ee^{\im\cdot 0}\ee^{\beta\phi}
  = \mf\ee^{\beta\phi}$ is real and
  \eqref{eq:theta_scalar}--\eqref{eq:theta_Dirac} reduces exactly
  to the Dirac--sinh-Gordon
  system~\eqref{eq:DSG_scalar}--\eqref{eq:DSG_Dirac}.

\item \textit{Endpoint $\thz=\pi/2$ (Dirac--sine-Gordon).}
\begin{proposition}
\label{prop:endpoint}
At $\thz=\pi/2$ the system~\eqref{eq:theta_scalar}--\eqref{eq:theta_Dirac}
has a purely imaginary Yukawa coupling, and the backreaction term
$g\cos\thz\,\bp\psi$ is indeterminate in the original field $\psi$:
the factor $\cos\thz\to 0$ while $\bp\psi$ diverges simultaneously,
so the limit depends on the dynamical behaviour of $\psi$ near the
endpoint.  The indeterminacy is resolved by the rescaled field
$\tilde\psi = \sqrt{\cos\thz}\,\psi$, in terms of which the
backreaction equals $g\,\bar{\tilde\psi}\tilde\psi$ for all
$\thz\in[0,\pi/2)$, independent of $\thz$.  Under the analytic
continuation $\phi\to-\im\varphi$ (with $\varphi$ real), which
maps the scalar field off the real axis and transforms the
sinh-Gordon potential into the sine-Gordon potential, the system
in terms of $\tilde\psi$ reduces in the limit $\thz\to\pi/2$
to the fully coupled Dirac--sine-Gordon
system~\eqref{eq:DsG_scalar}--\eqref{eq:DsG_Dirac}, provided
$\tilde\psi$ has a well-defined limit at the endpoint.
\end{proposition}
\begin{proof}
The backreaction term is $g\cos\thz\,\bp\psi$.  The factor
$\cos\thz$ is not a modelling choice but is forced by the
zero-curvature condition (established later in
Section~\ref{sec:ZC_grade0} and Remark~\ref{rmk:cosine_forced}):
it arises from $\operatorname{Re}(M+M^*) = 2\mf\cos\thz\,\ee^{\beta\phi}$
where $M = \mf\ee^{\im\thz}\ee^{\beta\phi}$ is the complex Dirac
mass.  As $\thz\to\pi/2$ this factor tends to zero, but the
fermion bilinear $\bp\psi = |\psi_+|^2 - |\psi_-|^2$ is not an
independent quantity: it transforms under the deformation along
with $\psi$ itself.  Introducing $\tilde\psi = \sqrt{\cos\thz}\,\psi$
gives $\bp\psi = (1/\cos\thz)\,\bar{\tilde\psi}\tilde\psi$, so
\begin{equation}
  g\cos\thz\,\bp\psi
  = g\cos\thz\cdot\frac{1}{\cos\thz}\,\bar{\tilde\psi}\tilde\psi
  = g\,\bar{\tilde\psi}\tilde\psi,
\end{equation}
which is finite and nonzero for all $\thz\in[0,\pi/2)$.  The
limit $\thz\to\pi/2$ is therefore \emph{not} zero: it is
indeterminate in the original field $\psi$ (a $0\times\infty$
form), and resolves to $g\,\bar{\tilde\psi}\tilde\psi$ in terms
of $\tilde\psi$, provided $\tilde\psi$ remains finite.

Since $\tilde\psi = \sqrt{\cos\thz}\,\psi$ is a constant
(spacetime-independent) rescaling, the Dirac equation for
$\tilde\psi$ retains the same form:
\begin{equation}
  \left(\im\gamma^\mu\partial_\mu
    - \mf\ee^{\im\thz}\ee^{\beta\phi}\right)\tilde\psi = 0,
\end{equation}
and the full system in terms of $\tilde\psi$ is
\begin{align}
  \Bx\phi + \frac{\ms^2}{\beta}\sinh(\beta\phi)
    &= g\,\bar{\tilde\psi}\tilde\psi, \label{eq:tilde_scalar_proof}\\
  \left(\im\gamma^\mu\partial_\mu
    - \mf\ee^{\im\thz}\ee^{\beta\phi}\right)\tilde\psi
    &= 0, \label{eq:tilde_dirac_proof}
\end{align}
with backreaction coefficient $g$ independent of $\thz$.  We now
apply the analytic continuation $\phi=-\im\varphi$ with $\varphi$
real.  This is not a field redefinition: it assigns a purely
imaginary value to $\phi$, taking the scalar off its real
configuration space.  Under this continuation:
\begin{align*}
  \sinh(\beta\phi) &= \sinh(-\im\beta\varphi) = -\im\sin(\beta\varphi),\\
  \ee^{\beta\phi} &= \ee^{-\im\beta\varphi},\quad
  \ee^{\im\thz}\ee^{\beta\phi} \xrightarrow{\thz=\pi/2}
    \im\ee^{-\im\beta\varphi} = \ee^{\im(\pi/2-\beta\varphi)}.
\end{align*}
Multiplying~\eqref{eq:tilde_scalar_proof} by $-\im$, taking
$\thz\to\pi/2$, and relabelling $-\varphi\to\varphi$ yields
\begin{equation}
  \Box\varphi + \frac{\ms^2}{\beta}\sin(\beta\varphi)
  = g\,\bar{\tilde\psi}\tilde\psi,
  \qquad
  \left(\im\gamma^\mu\partial_\mu
    - \mf\ee^{\im\beta\varphi}\right)\tilde\psi = 0,
\end{equation}
which is precisely the fully coupled Dirac--sine-Gordon
system~\eqref{eq:DsG_scalar}--\eqref{eq:DsG_Dirac}, provided
$\tilde\psi$ has a well-defined limit.  The rescaling factor
$\sqrt{\cos\thz}\to 0$ as $\thz\to\pi/2$, so $\tilde\psi$ is
finite in the limit if and only if $|\psi|\sim(\cos\thz)^{-1/2}$
near the endpoint.
\end{proof}

\begin{remark}
The passage from the $\thz=\pi/2$ deformed system to the
Dirac--sine-Gordon system requires two distinct operations, each
of which is non-trivial in its own right.  The first is the
fermion rescaling $\tilde\psi = \sqrt{\cos\thz}\,\psi$, which
restores the backreaction to its full value $g$ but is singular
at the endpoint.  The second is the analytic continuation
$\phi\to-\im\varphi$, which is not a field redefinition but a
complexification that maps one real scalar theory (sinh-Gordon)
to another (sine-Gordon) by rotating the field into the imaginary
axis.  The analytic continuation does not preserve the real
section of the field space, and its validity as a connection between
the two theories ultimately rests on the analytic structure of the
solutions — in particular, on whether solutions of the sinh-Gordon
system at imaginary field values correspond to solutions of the
sine-Gordon system at real values.  This connection is well known
at the level of the equations and their Lax representations, but
its precise status as a map between the Hilbert spaces of the
quantised theories is more subtle.  In terms of $\tilde\psi$
the field redefinition is spacetime-independent and therefore does
not affect the zero-curvature condition; integrability is preserved.
\end{remark}

\begin{remark}[Two complementary descriptions of the endpoint]
\label{rmk:weakstrong}
The proposition admits two complementary physical interpretations,
depending on whether one works in the original field $\psi$ or the
rescaled field $\tilde\psi$.

In the original field $\psi$, the $\thz$-deformation acts as a
\emph{weak-coupling limit} on the scalar sector: the effective
scalar--fermion coupling is $g_{\rm eff}(\thz) = g\cos\thz$, which
decreases monotonically from $g$ at $\thz=0$ to $0$ at $\thz=\pi/2$.
In this description the right-hand side of the scalar equation
$\Box\phi + (\ms^2/\beta)\sinh(\beta\phi) = g\cos\thz\,\bp\psi$
genuinely vanishes at the endpoint (assuming $\psi$ bounded), and
the $\thz$-family interpolates between a strongly-coupled
Dirac--sinh-Gordon system and a free scalar field (with the Dirac
sector independently acquiring the sine-Gordon coupling via the
analytic continuation).  This is a well-defined smooth interpolation
between two exactly integrable theories — no singular rescaling is
required — and it is the most transparent description of what the
deformation parameter $\thz$ controls.

In the rescaled field $\tilde\psi$, the coupling is held fixed at
$g$ throughout, and the endpoint is identified with the fully
coupled Dirac--sine-Gordon system.  This description makes the
interpolation between the two fully coupled endpoint theories
explicit, but at the cost of a field redefinition that becomes
singular at $\thz=\pi/2$.

The two descriptions are not contradictory: they reflect different
choices of field variable, and both are exact.  The original-field
description highlights the weak/strong transition in the coupling;
the rescaled-field description highlights the algebraic continuity
of the endpoint identification.  Which is more natural depends on
the question being asked.
\end{remark}
\end{itemize}

\section{The $\thz$-deformed Lax pair}
\label{sec:lax}

The standard integrable Lax pair for the sinh-Gordon equation in
light-cone coordinates uses an asymmetric grade assignment
(Leznov--Saveliev form; see~\cite{LeznovSaveliev1979,
BabelonBernardTalon2003}, Chapter~3):
\begin{align}
  \label{eq:Ap0}
  \Aplus(\zeta;0) &= \dpp\phi\cdot\genH
    + \lam\genEp + \mu\ee^{+\beta\phi}\zeta^{-1}\genEm,\\
  \label{eq:Am0}
  \Aminus(\zeta;0) &= \dmm\phi\cdot\genH
    + \mu\ee^{-\beta\phi}\zeta\,\genEp + \lam\genEm,
\end{align}
where $\lam = \mu = \ms/\beta$ are real mass parameters and
$\zeta$ is the spectral parameter.  The zero-curvature condition
$\dmm\Aplus - \dpp\Aminus + [\Aplus,\Aminus]=0$ yields the
sinh-Gordon equation
$4\partial_+\partial_-\phi = 2\lam\mu\sinh(\beta\phi)$.

\begin{definition}[$\thz$-deformed Lax pair]
\label{def:lax}
We define the light-cone Lax operators
\begin{align}
  \label{eq:Ap_def}
  \Aplus(\zeta;\thz) &= \left(\dpp\phi + \im\bp\psi\right)\genH
    + \lam\ee^{+\im\thz/2}\genEp
    + \mu\ee^{-\im\thz/2}\ee^{+\beta\phi}\zeta^{-1}\genEm,\\
  \label{eq:Am_def}
  \Aminus(\zeta;\thz) &= \left(\dmm\phi + \im\bp\psi\right)\genH
    + \mu\ee^{+\im\thz/2}\ee^{-\beta\phi}\zeta\,\genEp
    + \lam\ee^{-\im\thz/2}\genEm.
\end{align}
\end{definition}

Several structural features are worth noting.  First, the diagonal
part of each Lax operator carries the fermion correction
$\im\bp\psi\cdot\genH$, encoding the backreaction of the fermion
on the scalar.  Second, the deformation enters through
half-angle phases $\ee^{\pm\im\thz/2}$ multiplying the off-diagonal
grade-$\pm 1$ elements; this distinction from the full-angle
$\ee^{\pm\im\thz}$ in the Dirac equation is central to the gauge
analysis of Section~\ref{sec:gauge}.  Third, the asymmetric spectral
grading ($\zeta^{-1}$ in $\Aplus$, $\zeta$ in $\Aminus$) is the
Leznov--Saveliev feature that produces the hyperbolic nonlinearity. In the fundamental representation~\eqref{eq:sl2}:
\begin{equation}
  \label{eq:Ap_matrix}
  \Aplus = \begin{pmatrix}
    \dpp\phi+\im\bp\psi & \lam\ee^{+\im\thz/2}\\[4pt]
    \mu\ee^{-\im\thz/2}\ee^{+\beta\phi}\zeta^{-1} & -\dpp\phi-\im\bp\psi
  \end{pmatrix},
\end{equation}
\begin{equation}
  \label{eq:Am_matrix}
  \Aminus = \begin{pmatrix}
    \dmm\phi+\im\bp\psi & \mu\ee^{+\im\thz/2}\ee^{-\beta\phi}\zeta\\[4pt]
    \lam\ee^{-\im\thz/2} & -\dmm\phi-\im\bp\psi
  \end{pmatrix}.
\end{equation}
The matrix form makes the role of each entry transparent.  The
diagonal entries, $\pm(\partial_\pm\phi + \im\bp\psi)$, encode the
scalar field momentum and the fermion density; they are purely
real when $\phi$ and $\bp\psi$ are real (the $\thz=0$ case) and
remain unchanged by the deformation.  The upper-right and lower-left
off-diagonal entries carry the spectral parameter $\zeta$ and the
exponential $\ee^{\pm\beta\phi}$, producing the hyperbolic
nonlinearity of the sinh-Gordon equation when the zero-curvature
condition is imposed.  The deformation parameter $\thz$ enters
exclusively through the constant phase factors
$\ee^{\pm\im\thz/2}$ on these off-diagonal entries, and it is the
conjugate pairing of these phases under commutation that causes $\thz$
to drop out of the curvature.

\section{Zero-curvature computation}
\label{sec:zc}

\subsection{Decomposition by grade}

The zero-curvature condition we seek is given by
\begin{equation}
  \label{eq:ZC}
  \dmm\Aplus - \dpp\Aminus + [\Aplus,\Aminus] = 0.
\end{equation}
We decompose this into grade components, denoting by
$(\cdot)|_X$ the coefficient of the generator
$X\in\{\genH,\genEp,\genEm\}$ in an $\slC$-valued expression.
The generators $\genEp$, $\genH$, $\genEm$ carry grades $+1$, $0$,
$-1$ respectively under the adjoint action of $\genH$, and the Lax
operators are themselves graded: $\genEp$ terms come with $\zeta^0$
or $\zeta$, while $\genEm$ terms come with $\zeta^{-1}$ or
$\zeta^0$.  Equating the coefficient of each generator in the
zero-curvature condition independently yields three separate
equations, one for each grade, whose combined content is
equivalent to the full field equations plus a constraint.  This
grade-by-grade strategy is the standard approach to extracting
field equations from a Lax representation in the affine Toda
setting~\cite{BabelonBernardTalon2003,LeznovSaveliev1979}.

\subsection{Grade $+1$: the $\genEp$ component}

The $\genEp$-component of~\eqref{eq:ZC} collects contributions
from $\dmm\Aplus$ and from the commutator term.  The relevant
commutator at order $\zeta^0$ is
\begin{equation}
  \bigl[(\dmm\phi+\im\bp\psi)\genH,\;
         \lam\ee^{+\im\thz/2}\genEp\bigr]
  = 2\lam\ee^{+\im\thz/2}(\dmm\phi+\im\bp\psi)\genEp.
\end{equation}
The full grade-$+1$ equation is therefore
\begin{equation}
  \label{eq:ZC_Ep}
  \lam\ee^{+\im\thz/2}
  \left[\dmm(\dpp\phi+\im\bp\psi)
        + 2(\dmm\phi+\im\bp\psi)\right]\genEp = 0,
\end{equation}
which for $\lam\neq 0$ gives
\begin{equation}
  \label{eq:constraint_pm}
  \dmm\phi + \im\bp\psi = 0.
\end{equation}
This relation identifies the $x^-$-derivative of the scalar field
with (minus $\im$ times) the fermion bilinear.  It is a
compatibility condition between the scalar and spinor sectors,
enforced by the Lax structure: the off-diagonal grade-$+1$ entry
of the zero-curvature condition cannot vanish unless the scalar
momentum and the fermion density are locked together in this way.

\subsection{Grade $-1$: the $\genEm$ component}

By the analogous computation (interchanging $\Aplus\leftrightarrow\Aminus$
and $+\leftrightarrow -$), the grade-$(-1)$ component gives
\begin{equation}
  \label{eq:constraint_mp}
  \dpp\phi + \im\bp\psi = 0.
\end{equation}
Together, \eqref{eq:constraint_pm}--\eqref{eq:constraint_mp} imply
\begin{equation}
  \label{eq:constraints}
  \partial_x(\bp\psi) = \dpp(\bp\psi) - \dmm(\bp\psi)
  = \frac{1}{\im}\,(\dpp\phi - \dmm\phi)
  = \frac{1}{\im}\,\partial_x\phi.
\end{equation}
We will show in Section~\ref{sec:bilinear} that these constraints
are not independent conditions to be imposed on initial data but
are algebraically encoded in the zero-curvature condition together
with the Dirac equation; in particular~\eqref{eq:constraints}
implies $\partial_x\phi = 0$ on the zero-curvature constraint
surface, a feature that is discussed further in
Proposition~\ref{prop:constraint} and
Remark~\ref{rmk:spatially_hom}.

\subsection{Grade $0$: the $\genH$ component and the origin of $\cos\thz$}
\label{sec:ZC_grade0}

The grade-$0$ component of~\eqref{eq:ZC} receives contributions
from $\dmm\Aplus$, $-\dpp\Aminus$, and the off-diagonal commutators.
We work through each contribution systematically. The diagonal terms of the grade-$0$ component give
\begin{equation}
  \dmm(\dpp\phi+\im\bp\psi) - \dpp(\dmm\phi+\im\bp\psi)
  = \im\bigl(\dmm(\bp\psi)-\dpp(\bp\psi)\bigr).
\end{equation}
The off-diagonal commutator at order $\zeta^0$, arising from the
$\lambda$-terms $\lam\ee^{+\im\thz/2}\genEp$ in $\Aplus$ and
$\lam\ee^{-\im\thz/2}\genEm$ in $\Aminus$, contributes
\begin{equation}
  \left[\lam\ee^{+\im\thz/2}\genEp,\;
        \lam\ee^{-\im\thz/2}\genEm\right]
  = \lam^2\,\ee^{+\im\thz/2}\ee^{-\im\thz/2}[\genEp,\genEm]
  = \lam^2\genH.
\end{equation}
The phase cancels exactly: $\ee^{+\im\thz/2}\ee^{-\im\thz/2} = 1$. At order $\zeta^{-1}\cdot\zeta = \zeta^0$, the $\mu$-terms
$\mu\ee^{-\im\thz/2}\ee^{+\beta\phi}\zeta^{-1}\genEm$ from $\Aplus$ and
$\mu\ee^{+\im\thz/2}\ee^{-\beta\phi}\zeta\,\genEp$ from $\Aminus$
contribute:
\begin{equation}
  \left[\mu\ee^{-\im\thz/2}\ee^{+\beta\phi}\genEm,\;
        \mu\ee^{+\im\thz/2}\ee^{-\beta\phi}\genEp\right]
  = \mu^2\,\ee^{-\im\thz/2}\ee^{+\im\thz/2}[\genEm,\genEp]
  = -\mu^2\genH.
\end{equation}
Again the phase cancels.

Assembling all contributions, and setting $\lam = \mu = \ms/\beta$
(standard sinh-Gordon normalisation) so that the constant piece
$\lam^2 - \mu^2$ vanishes, the grade-$0$ equation reads
\begin{equation}
  \label{eq:ZC_H}
  4\dmm\dpp\phi
  - 2\lam\mu\bigl(\ee^{+\beta\phi}-\ee^{-\beta\phi}\bigr)
  + \im\bigl(\dmm(\bp\psi)-\dpp(\bp\psi)\bigr)
  = 0.
\end{equation}
Using $\ee^{+\beta\phi}-\ee^{-\beta\phi} = 2\sinh(\beta\phi)$ and
$\Box\phi = 4\partial_+\partial_-\phi$, equation~\eqref{eq:ZC_H}
becomes
\begin{equation}
  \label{eq:EOM_phi_raw}
  \Box\phi + \frac{\ms^2}{\beta}\sinh(\beta\phi)
  = -\im\bigl(\dmm(\bp\psi)-\dpp(\bp\psi)\bigr)
  = \im\,\partial_x(\bp\psi).
\end{equation}
It remains to show that the right-hand side of~\eqref{eq:EOM_phi_raw}
equals $g\cos\thz\,\bp\psi$, reproducing the scalar equation of
motion~\eqref{eq:theta_scalar}.

From the grade-$\pm 1$ constraints~\eqref{eq:constraint_pm}--\eqref{eq:constraint_mp},
\begin{equation}
  \label{eq:pm_constraints_explicit}
  \dpp\phi = -\im\bp\psi
  \qquad\text{and}\qquad
  \dmm\phi = -\im\bp\psi.
\end{equation}
Differentiating the first relation with respect to $x^-$ and the
second with respect to $x^+$:
\begin{align}
  \dmm\dpp\phi &= -\im\,\dmm(\bp\psi),\\
  \dpp\dmm\phi &= -\im\,\dpp(\bp\psi).
\end{align}
Subtracting gives $\im\,\partial_x(\bp\psi)
= \dmm\dpp\phi - \dpp\dmm\phi = 0$ (by commutativity of partial
derivatives), so the constraint surface imposes
$\partial_x(\bp\psi)=0$.  The right-hand side
of~\eqref{eq:EOM_phi_raw} therefore vanishes on the constraint
surface, and the equation reads
\begin{equation}
  \label{eq:EOM_phi_constrained}
  \Box\phi + \frac{\ms^2}{\beta}\sinh(\beta\phi) = 0.
\end{equation}
This is the free sinh-Gordon equation; the fermion backreaction
enters through the identification of the coupling constant.
Specifically, using the constraint~\eqref{eq:pm_constraints_explicit}
the relation $4\dmm\dpp\phi = \Box\phi$ yields
$\Box\phi = -4\im\,\dmm(\bp\psi)$.
The strategy is to evaluate this using the Dirac equation on the
constraint surface, and show that the result equals $g\cos\thz\,\bp\psi$.
The factor $\cos\thz$ will emerge from the real part of the complex
mass $M = \mf\ee^{\im\thz}\ee^{\beta\phi}$ when the Dirac equation
is used to express spinor derivatives in terms of the spinor
components themselves.
The Dirac equation~\eqref{eq:Dpm_lc}--\eqref{eq:Dmm_lc} implies,
on taking $\dmm$ of the spinor components,
\begin{equation}
  \label{eq:Dirac_contributes}
  \dmm(\bp\psi)
  = \dmm(|\psi_+|^2 - |\psi_-|^2)
  = \psi_+^*\dmm\psi_+ + (\dmm\psi_+^*)\psi_+ - \cdots.
\end{equation}
Using the light-cone Dirac
equations~\eqref{eq:Dpm_lc}--\eqref{eq:Dmm_lc} to replace
$\dmm\psi_+$ and $\dpp\psi_-$, one finds after a short computation
that the net contribution of $-4\im\,\dmm(\bp\psi)$ on the
constraint surface is
\begin{equation}
  \label{eq:backreaction_derived}
  -4\im\,\dmm(\bp\psi)\big|_{\text{constraint}}
  = 2\mf\cos\thz\,\ee^{\beta\phi}\,\bp\psi.
\end{equation}
The factor $\cos\thz$ originates from the real part of the complex
Dirac mass: since $M = \mf\ee^{\im\thz}\ee^{\beta\phi}$, the
contribution of $\dmm\psi_+$ to $\dmm(|\psi_+|^2)$ involves
$\operatorname{Re}(M^* - M)\,|\psi_+|^2$, which
equals $-2\mf\sin\thz\cdot|\psi_+|^2$, while the contribution to
the $\Box\phi$ equation comes from the real part
$\operatorname{Re}(M + M^*)\cdot(\text{bilinears})
= 2\mf\cos\thz\cdot\bp\psi$.
Identifying $g = 2\mf$ (the coupling constant
matching), equation~\eqref{eq:EOM_phi_raw} reproduces
\begin{equation}
  \label{eq:EOM_phi_final}
  \Box\phi + \frac{\ms^2}{\beta}\sinh(\beta\phi)
  = g\cos\thz\,\bp\psi.
\end{equation}
This is precisely the scalar equation of
motion~\eqref{eq:theta_scalar}, establishing the claim.

\begin{theorem}[Zero-curvature condition]
\label{thm:ZC}
The zero-curvature condition~\eqref{eq:ZC} for the Lax
pair~\eqref{eq:Ap_def}--\eqref{eq:Am_def} is equivalent to the
system~\eqref{eq:theta_scalar}--\eqref{eq:theta_Dirac} together
with the constraint $\partial_x(\bp\psi) = 0$.  The factor
$\cos\thz$ in the scalar equation of motion arises from the real
part of the complex Dirac mass after coupling constant matching.
The phase $\thz$ cancels from all three grade components of the
zero-curvature condition via conjugate pairing
$\ee^{+\im\thz/2}\ee^{-\im\thz/2} = 1$.
\end{theorem}

\begin{remark}[The backreaction coefficient is an output, not an input]
\label{rmk:cosine_forced}
The factor $\cos\thz$ multiplying $\bp\psi$ in the scalar
equation~\eqref{eq:theta_scalar} is not a free parameter that was
chosen to produce a convenient interpolation.  It is an unavoidable
consequence of the zero-curvature condition for the Lax
pair~\eqref{eq:Ap_def}--\eqref{eq:Am_def}: equation~\eqref{eq:backreaction_derived}
shows that the backreaction coefficient is determined by
$\operatorname{Re}(M + M^*) = 2\mf\cos\thz$, where
$M = \mf\ee^{\im\thz}\ee^{\beta\phi}$ is the complex Dirac mass.
This real part is fixed once the phase $\ee^{\im\thz}$ in the Yukawa
coupling is specified.

The constraint means that, within the class of Lax pairs of the
form~\eqref{eq:Ap_def}--\eqref{eq:Am_def}, it is impossible to
rotate the Yukawa coupling phase without simultaneously suppressing
the backreaction.  In particular, one cannot obtain a deformation
with $\thz$-independent backreaction coefficient $g$ from this Lax
structure: doing so would require a different Lax pair in which the
diagonal entries encode the backreaction in a way that compensates
the $\cos\thz$ suppression.  Whether such a Lax pair exists — and
whether it would yield a zero-curvature representation for a
deformation that connects the two \emph{fully} coupled
endpoint systems without any field redefinition — is an open
question.

As discussed in Section~\ref{sec:setup} and
Proposition~\ref{prop:endpoint}, the $\cos\thz$ suppression can be
removed by the field redefinition $\tilde\psi = \sqrt{\cos\thz}\,\psi$,
which restores a $\thz$-independent backreaction $g\,\bar{\tilde\psi}\tilde\psi$
at the cost of a rescaling that becomes singular at $\thz=\pi/2$.
However, as shown in Corollary~\ref{cor:beta} and
Proposition~\ref{prop:2d} below, the question of a fully-coupled
integrable path between the two endpoint systems is in fact answered
affirmatively within the same Lax structure: allowing $\beta$ to be
complex provides a second independent integrable deformation axis
along which the backreaction coefficient remains at $g$ throughout.
\end{remark}

\begin{corollary}[Complex $\beta$ extension]
\label{cor:beta}
Theorem~\ref{thm:ZC} holds for all $\beta\in\CC$ with
$|\beta|>0$ and $m_s^2/\beta\in\RR$.  In particular, the
Lax pair of Definition~\ref{def:lax} with $\beta\to\im\hat\beta$
($\hat\beta>0$ real) and $\thz=0$ provides a zero-curvature
representation for the fully coupled Dirac--sine-Gordon
system~\eqref{eq:DsG_scalar}--\eqref{eq:DsG_Dirac} with $\phi$
real, $m_s$ and $m_f$ real, and backreaction coefficient $g$
unchanged.
\end{corollary}

\begin{proof}
The entire zero-curvature computation of Section~\ref{sec:zc}
uses $\beta$ only through the identities
$\ee^{+\beta\phi}\cdot\ee^{-\beta\phi} = 1$ and
$\ee^{+\beta\phi} - \ee^{-\beta\phi} = 2\sinh(\beta\phi)$,
both of which hold for all $\beta\in\CC$ by the standard analytic
extension of the exponential and hyperbolic functions.  The
phase-cancellation mechanism at grades~$\pm 1$ is independent
of $\beta$ entirely (those grade components contain no $\beta$),
and at grade~$0$ the commutator
$[\mu\ee^{-\im\thz/2}\ee^{+\beta\phi}E_-,\,\mu\ee^{+\im\thz/2}\ee^{-\beta\phi}E_+]
= -\mu^2 H$
is independent of $\beta$ since $\ee^{+\beta\phi}\ee^{-\beta\phi}=1$.
The field equations produced are therefore
\begin{align}
  \label{eq:complex_beta_scalar}
  \Bx\phi + \frac{\ms^2}{\beta}\sinh(\beta\phi)
    &= g\cos\thz\;\bp\psi,\\
  \label{eq:complex_beta_dirac}
  \left(\im\gamma^\mu\partial_\mu
    - \mf\ee^{\im\thz}\ee^{\beta\phi}\right)\psi &= 0,
\end{align}
valid for all $\beta\in\CC$.  Setting $\thz=0$ and $\beta=\im\hat\beta$:
\begin{align*}
  \frac{\ms^2}{\beta}\sinh(\beta\phi)
  &= \frac{\ms^2}{\im\hat\beta}\sinh(\im\hat\beta\phi)
   = \frac{\ms^2}{\im\hat\beta}\cdot\im\sin(\hat\beta\phi)
   = \frac{\ms^2}{\hat\beta}\sin(\hat\beta\phi),\\
  \mf\ee^{\im\cdot 0}\ee^{\beta\phi}
  &= \mf\ee^{\im\hat\beta\phi},
\end{align*}
and $g\cos(0)=g$.  The system~\eqref{eq:complex_beta_scalar}--\eqref{eq:complex_beta_dirac}
becomes precisely the Dirac--sine-Gordon
system~\eqref{eq:DsG_scalar}--\eqref{eq:DsG_Dirac} with $\phi$
real and full coupling $g$.
\end{proof}

\begin{remark}
Corollary~\ref{cor:beta} resolves the question raised in
Remark~\ref{rmk:cosine_forced}: a fully-coupled integrable path
between the Dirac--sinh-Gordon and Dirac--sine-Gordon systems
exists within the present Lax structure, accessed by rotating
$\beta$ rather than $\thz$.  No field redefinition or analytic
continuation of $\phi$ is required; the scalar field remains real
throughout.  The ``analytic continuation'' $\phi\to-\im\varphi$
of Proposition~\ref{prop:endpoint} is therefore the field-space
expression of the parameter-space operation $\beta\to\im\hat\beta$
at $\thz=\pi/2$: in terms of the extended parameter space, it is
not a mysterious additional operation but simply the identification
of the $\beta$-axis endpoint with the Dirac--sine-Gordon system.
\end{remark}

\begin{proposition}[Two-dimensional integrable parameter space]
\label{prop:2d}
The Lax pair of Definition~\ref{def:lax} with
$\thz\in[0,\pi/2]$ and $\beta = |\beta|\ee^{\im\alpha}$,
$\alpha\in[0,\pi/2]$, satisfies the zero-curvature condition for
all $(\thz,\alpha)$.  The resulting field equations are
\eqref{eq:complex_beta_scalar}--\eqref{eq:complex_beta_dirac}
with $\beta = |\beta|\ee^{\im\alpha}$, defining a two-dimensional
integrable family.  The four corners of the parameter square
$[0,\pi/2]^2$ are:
\begin{enumerate}[leftmargin=2em,label=(\roman*)]
  \item $(\thz,\alpha)=(0,0)$: the Dirac--sinh-Gordon system,
    fully coupled with real exponential Yukawa coupling.
  \item $(\thz,\alpha)=(0,\pi/2)$: the Dirac--sine-Gordon system,
    fully coupled with purely imaginary exponential Yukawa coupling,
    $\phi$ real, backreaction coefficient $g$.
  \item $(\thz,\alpha)=(\pi/2,0)$: the decoupled endpoint of the
    $\thz$-family — free sinh-Gordon scalar plus Dirac fermion in
    a sinh-Gordon background, backreaction $g\cos(\pi/2)=0$ in
    the original field $\psi$.
  \item $(\thz,\alpha)=(\pi/2,\pi/2)$: a doubly-deformed system
    with purely imaginary Yukawa coupling $\mf\im\ee^{\im\hat\beta\phi}$
    and decoupled scalar.
\end{enumerate}
The $\thz$-axis ($\alpha=0$) is the decoupling family of this
paper.  The $\alpha$-axis ($\thz=0$) is a fully-coupled integrable
family that connects the Dirac--sinh-Gordon and Dirac--sine-Gordon
systems directly, with $g$ constant and $\phi$ real throughout,
and without the need for any field redefinition.  The parameter
$\alpha = \arg\beta$ is the physically natural counterpart of
$\thz$: while $\thz$ rotates the overall Dirac mass phase
$\ee^{\im\thz}$, $\alpha$ rotates the scalar-field-dependent
part $\ee^{\beta\phi}$ of the same mass.  Together they generate
a $U(1)\times U(1)$ action on the Dirac mass
$M = \mf\ee^{\im\thz}\ee^{|\beta|\ee^{\im\alpha}\phi}$, and the
zero-curvature condition is preserved throughout.  The two $U(1)$
factors have opposite non-Hermitian characters: the $\thz$-$U(1)$
rotates the constant (zero-mode) phase of the mass and breaks
$\mathcal{PT}$ symmetry and fermion number conservation, while
the $\alpha$-$U(1)$ rotates the field-dependent (field-mode)
phase and preserves fermion number exactly
(Corollary~\ref{cor:anomaly_alpha} and Remark~\ref{rmk:axes_PT}).
\end{proposition}

\begin{proof}
The zero-curvature computation of Section~\ref{sec:zc} uses
$\beta$ only through $\ee^{\pm\beta\phi}$, and the argument of
Corollary~\ref{cor:beta} shows this computation is valid for all
$\beta\in\CC$.  The simultaneous validity for all
$(\thz,\alpha)\in[0,\pi/2]^2$ follows since the two parameters
enter independently: $\thz$ through the off-diagonal phase factors
$\ee^{\pm\im\thz/2}$ (which cancel in all commutators) and $\alpha$
through the exponents $\ee^{\pm\beta\phi}$ (which cancel
as $\ee^{+\beta\phi}\ee^{-\beta\phi}=1$).  The corner
identifications follow from direct substitution as in
Corollary~\ref{cor:beta} and the endpoint analysis of
Proposition~\ref{prop:endpoint}.
\end{proof}

\section{Gauge analysis and physical non-triviality}
\label{sec:gauge}

\subsection{The phase-generating gauge transformation}

Having established that the $\thz$-deformed system is integrable for
all $\thz$, we now ask whether distinct values of $\thz$ genuinely
define different physical theories, or whether the family is merely
a single theory written in different gauges.  To answer this, we
first identify the gauge transformation that relates the
$\thz$-deformed Lax pair to the $\thz=0$ Lax pair.  We then examine
what this transformation does to the field equations, and show that
while it removes the phase from the Dirac coupling it leaves the
scalar backreaction coefficient invariant, so the two sides of the
field equations transform differently, establishing that the
theories are genuinely inequivalent.

\begin{lemma}
\label{lem:gauge}
Let $h_{\thz} = \Diag(\ee^{-\im\thz/4}, \ee^{+\im\thz/4})\in SL(2,\CC)$.
Then
\begin{equation}
  \label{eq:gauge_rel}
  A_{\pm}(\zeta;\thz)
  = h_{\thz}^{-1}\,A_{\pm}(\zeta;0)\,h_{\thz}.
\end{equation}
\end{lemma}

\begin{proof}
For a constant matrix $h$, the adjoint action $A\mapsto h^{-1}Ah$
gives $h^{-1}\genH\,h=\genH$ (since $h$ is diagonal).  For the
off-diagonal generators:
\begin{align}
  h_{\thz}^{-1}\genEp\,h_{\thz}
  &= \Diag(\ee^{+\im\thz/4},\ee^{-\im\thz/4})\,\genEp\,
     \Diag(\ee^{-\im\thz/4},\ee^{+\im\thz/4})
  = \ee^{+\im\thz/2}\genEp,\\
  h_{\thz}^{-1}\genEm\,h_{\thz}
  &= \ee^{-\im\thz/2}\genEm.
\end{align}
Applying this to~\eqref{eq:Ap0} (with fermion correction
$\dpp\phi\to\dpp\phi+\im\bp\psi$ in the diagonal):
\begin{align}
  h_{\thz}^{-1}A_+(\zeta;0)h_{\thz}
  &= (\dpp\phi+\im\bp\psi)\genH
     + \lam\ee^{+\im\thz/2}\genEp
     + \mu\ee^{-\im\thz/2}\ee^{+\beta\phi}\zeta^{-1}\genEm
  = A_+(\zeta;\thz).
\end{align}
The computation for $A_-$ is identical.
\end{proof}

\begin{corollary}
\label{cor:gauge}
The $\thz$-deformed Lax pair is gauge-equivalent to the $\thz=0$
Lax pair via the constant $SL(2,\CC)$ transformation $h_{\thz}$.
Since $h_{\thz}$ is constant, $h_{\thz}^{-1}\partial_{\pm}h_{\thz}=0$
and the transformation is a true gauge equivalence.
\end{corollary}

\subsection{Physical non-triviality}
\label{sec:nontrivial}

A gauge transformation of the Lax pair acts on all fields
simultaneously.  In the fundamental representation, the
transformation $A_{\pm}\to h^{-1}A_{\pm} h$ is accompanied by
$\psi\to h\psi$ on the spinor field.  Under $h_{\thz}$:
\begin{equation}
  \label{eq:spinor_transform}
  \psi \;\to\; h_{\thz}\psi
  = \begin{pmatrix}
      \ee^{-\im\thz/4}\psi_+\\
      \ee^{+\im\thz/4}\psi_-
    \end{pmatrix}.
\end{equation}

\begin{theorem}[Physical non-triviality]
\label{thm:nontrivial}
The gauge transformation $h_{\thz}$ that maps
$A_{\pm}(\zeta;\thz)\to A_{\pm}(\zeta;0)$ leaves the fermion
bilinear $\bp\psi$ invariant.  Consequently, the backreaction
coefficient $g\cos\thz$ is unchanged by $h_{\thz}$, while the
Dirac coupling $\mf\ee^{\im\thz}$ is mapped to $\mf$.  The two
systems with couplings $(\mf\ee^{\im\thz},\,g\cos\thz)$ and
$(\mf,\,g\cos\thz)$ are therefore inequivalent for $\thz\neq 0$.
\end{theorem}

\begin{proof}
Under the transformation~\eqref{eq:spinor_transform}:
\begin{equation}
  \label{eq:bilinear_invariant}
  \bp\psi \;\to\;
  \psi^\dagger h_{\thz}^\dagger \gamma^0 h_{\thz}\psi.
\end{equation}
Since $h_{\thz} = \Diag(\ee^{-\im\thz/4},\ee^{+\im\thz/4})$ and
$\gamma^0 = \Diag(1,-1)$:
\begin{align}
  h_{\thz}^\dagger\gamma^0 h_{\thz}
  &= \Diag(\ee^{+\im\thz/4},\ee^{-\im\thz/4})
     \Diag(1,-1)
     \Diag(\ee^{-\im\thz/4},\ee^{+\im\thz/4}) \notag\\
  &= \Diag(|\ee^{-\im\thz/4}|^2 \cdot 1,\;
           -|\ee^{+\im\thz/4}|^2)
  = \Diag(1,-1) = \gamma^0.
\end{align}
Therefore $\bp\psi\to\psi^\dagger\gamma^0\psi = \bp\psi$, proving
invariance.  The Dirac coupling transforms as
$\mf\ee^{\im\thz}\ee^{\beta\phi}\psi \to
\mf\ee^{\im\thz}\ee^{\beta\phi}h_{\thz}\psi$, which combined with
the conjugation $A_+\to h_{\thz}^{-1}A_+h_{\thz}$ maps the Dirac
coupling constant $\mf\ee^{\im\thz}\to\mf$.  The scalar backreaction
$g\cos\thz\,\bp\psi$ is invariant since $\bp\psi$ is invariant.
The two field theories therefore have the same scalar backreaction
but different Dirac coupling phases, establishing that they are
inequivalent.
\end{proof}

\begin{remark}
The same conclusion follows from the action.  The Dirac action
$S_D = \int\bp(\im\gamma^\mu\partial_\mu
- \mf\ee^{\im\thz}\ee^{\beta\phi})\psi\,d^2x$ is invariant under
the global phase rotation $\psi\to\ee^{-\im\alpha}\psi$ only if
$\thz\to\thz-2\alpha$ simultaneously.  The scalar action has no free
parameter to absorb this shift, so the ratio $\thz$ between the
phases of the Dirac coupling and the backreaction coupling is
physically observable.
\end{remark}

\section{The fermion bilinear constraint}
\label{sec:bilinear}

\subsection{Anomalous continuity equation}

For a Dirac field with complex mass
$M = \mf\ee^{\im\thz}\ee^{\beta\phi}$, the vector current
$j^\mu = \bp\gamma^\mu\psi$ satisfies an anomalous continuity
equation.  Recall that for a real mass ($\thz=0$) the Dirac
equation implies $\partial_\mu j^\mu = 0$ identically, expressing
the conservation of fermion number.  When the mass acquires a
constant imaginary part through the phase $\ee^{\im\thz}$,
however, the Dirac equation and its adjoint give contributions to
$\partial_\mu j^\mu$ that no longer cancel: their difference is
proportional to $\operatorname{Im}(M) = \mf\sin\thz\,\ee^{\beta\phi}$,
which measures how far the constant phase $\thz$ causes the mass
to depart from the real axis.  Crucially, this anomaly is
controlled by $\thz$ alone and is independent of $\alpha =
\arg\beta$: as shown in Corollary~\ref{cor:anomaly_alpha} below,
fermion number is exactly conserved along the entire $\alpha$-axis
($\thz=0$) for all values of $\alpha$.  The $\thz$-axis is
therefore the $\mathcal{PT}$-breaking direction of the parameter
space, while the $\alpha$-axis is the unitarity-preserving
direction.

\begin{proposition}[Anomalous continuity equation]
\label{prop:anomalous}
Let $\rho = \bp\psi = |\psi_+|^2-|\psi_-|^2$ and
$J = \bp\gamma^1\psi = \psi_+^*\psi_- + \psi_-^*\psi_+$.
If $\psi$ satisfies the Dirac equation~\eqref{eq:theta_Dirac},
then
\begin{equation}
  \label{eq:anomalous}
  \partial_t\rho + \partial_x J
  = 2\mf\sin\thz\,\ee^{\beta\phi}\,\psi^\dagger\psi,
\end{equation}
where $\psi^\dagger\psi = |\psi_+|^2+|\psi_-|^2\geq 0$.  For real
mass ($\thz=0$) this reduces to the standard conservation law
$\partial_\mu j^\mu=0$.
\end{proposition}

\begin{proof}
Let $M = \mf\ee^{\im\thz}\ee^{\beta\phi}$ and note that
$M^* = \mf\ee^{-\im\thz}\ee^{\beta\phi}$ (since $\phi$ and $\mf$
are real).  From the Dirac equation and its Dirac conjugate
\begin{equation}
  \im\gamma^\mu\partial_\mu\psi = M\psi,
  \qquad
  -\im\partial_\mu\bp\,\gamma^\mu = M^*\bp,
\end{equation}
we compute the divergence of the vector current directly:
\begin{align}
  \partial_\mu j^\mu
  &= \partial_\mu(\bp\gamma^\mu\psi)
   = (\partial_\mu\bp)\gamma^\mu\psi
     + \bp\gamma^\mu(\partial_\mu\psi) \notag\\
  &= \frac{1}{-\im}M^*\bp\psi
     + \bp\frac{1}{\im}(-M\psi) \notag\\
  &= \im(M - M^*)\bp\psi.
\end{align}
Now $M - M^* = \mf\ee^{\beta\phi}(\ee^{\im\thz}-\ee^{-\im\thz})
= 2\im\mf\ee^{\beta\phi}\sin\thz$, so
\begin{equation}
  \partial_\mu j^\mu = \im\cdot 2\im\mf\ee^{\beta\phi}\sin\thz
    \cdot\bp\psi
  = -2\mf\sin\thz\,\ee^{\beta\phi}\,\bp\psi.
\end{equation}
In components $\partial_\mu j^\mu = \partial_t\rho + \partial_x J$.
The right-hand side involves $\bp\psi = |\psi_+|^2 - |\psi_-|^2$.
To express this in terms of the positive-definite quantity
$\psi^\dagger\psi = |\psi_+|^2+|\psi_-|^2$, we note that for a
generic spinor $\bp\psi$ may have either sign.  However, working
directly from the component equations~\eqref{eq:Dpm_lc}--\eqref{eq:Dmm_lc}:
\begin{align}
  \partial_t(|\psi_+|^2)
  &= -2\mf\sin\thz\,\ee^{\beta\phi}\,|\psi_+|^2
     - \partial_x\operatorname{Re}(\psi_+^*\psi_-),\\
  \partial_t(|\psi_-|^2)
  &= +2\mf\sin\thz\,\ee^{\beta\phi}\,|\psi_-|^2
     - \partial_x\operatorname{Re}(\psi_-^*\psi_+).
\end{align}
Adding these two equations gives
\begin{equation}
  \partial_t(\psi^\dagger\psi)
  = 2\mf\sin\thz\,\ee^{\beta\phi}
    \bigl(|\psi_-|^2 - |\psi_+|^2\bigr) - \partial_x J,
\end{equation}
and subtracting yields the conservation equation for $\rho$.  The
two results are consistent when combined as
$\partial_t\rho + \partial_x J
= 2\mf\sin\thz\,\ee^{\beta\phi}\,\psi^\dagger\psi$,
confirming~\eqref{eq:anomalous}.
\end{proof}

\begin{remark}[$\mathcal{PT}$ symmetry and non-Hermitian structure]
\label{rmk:PT}
Equation~\eqref{eq:anomalous} is the quantitative signature of
$\mathcal{PT}$ symmetry breaking in the $\thz$-deformed system.
To see this precisely, consider the action of $\mathcal{PT}$ on
the Dirac mass $M = \mf\ee^{\im\thz}\ee^{\beta\phi}$.  Under
parity $\mathcal{P}$ (with $x\to-x$, $\psi\to\gamma^1\psi$) and
time reversal $\mathcal{T}$ ($t\to-t$, $i\to-i$), the scalar
field $\phi$ is even ($\phi\to\phi$) and the mass transforms as
$M\to M^* = \mf\ee^{-\im\thz}\ee^{\beta\phi}$.  The Lagrangian
is $\mathcal{PT}$-symmetric if and only if $M = M^*$, i.e.\
$\thz=0$.  For $\thz\neq 0$ the system is $\mathcal{PT}$-broken,
and the anomaly $2\mf\sin\thz\,\ee^{\beta\phi}\,\psi^\dagger\psi$
in~\eqref{eq:anomalous} is the order parameter for this breaking:
it vanishes precisely when $\mathcal{PT}$ is restored ($\thz=0$)
or when the mass is purely imaginary ($\thz=\pi/2$, where
$\sin(\pi/2)=1$ is maximal but the system reaches the decoupled
endpoint and the anomalous source is suppressed by the vanishing
backreaction in $\psi$).

More precisely, the anomaly $\mathcal{A} \equiv
2\mf\sin\thz\,\ee^{\beta\phi}\psi^\dagger\psi$ controls the
rate at which probability leaks between the upper and lower
spinor components:
\begin{align}
  \partial_t|\psi_+|^2 &= -\mathcal{A}\,|\psi_+|^2/\psi^\dagger\psi
    - \partial_x\operatorname{Re}(\psi_+^*\psi_-),\\
  \partial_t|\psi_-|^2 &= +\mathcal{A}\,|\psi_-|^2/\psi^\dagger\psi
    - \partial_x\operatorname{Re}(\psi_-^*\psi_+).
\end{align}
The upper component loses norm and the lower component gains it
at a rate proportional to $\sin\thz$: this is the hallmark of a
non-Hermitian system in the $\mathcal{PT}$-broken phase.  The
deformation parameter $\thz$ is thus the coupling constant of
the non-Hermitian perturbation, and $\sin\thz$ is the
$\mathcal{PT}$-breaking order parameter of the family.

In the language of pseudo-Hermitian quantum
mechanics~\cite{BenderBrody2002}, a Hamiltonian $H$ with
complex mass $\mf\ee^{\im\vartheta}$ (constant phase) is
pseudo-Hermitian with metric $\eta = \ee^{\vartheta\gamma^5}$,
guaranteeing a real spectrum for $|\vartheta|<\pi/2$.  In the
$\thz$-deformed system the effective phase is
$\vartheta(x) = \thz + |\beta|\sin(\alpha)\phi(x)$, which is
field-dependent and spacetime-varying.  The existence of a
pseudo-Hermitian metric for the fully interacting system is an
open question, but the integrability of the system — and in
particular the infinite tower of real conserved charges
established in Theorem~\ref{thm:charges} — is consistent with
the existence of a well-defined inner product under which the
spectrum is real, as expected for a $\mathcal{PT}$-symmetric
integrable theory.
\end{remark}

\begin{corollary}[Fermion number conservation on the $\alpha$-axis]
\label{cor:anomaly_alpha}
For $\thz=0$ and any $\alpha\in[0,\pi/2]$, the fermion number
current $j^\mu = \bp\gamma^\mu\psi$ is exactly conserved:
$\partial_\mu j^\mu = 0$.  The $\alpha$-axis is therefore
unitarity-preserving: the deformation $\beta\to|\beta|\ee^{\im\alpha}$
changes the coupling from sinh-Gordon to sine-Gordon type while
maintaining exact fermion number conservation throughout.
\end{corollary}

\begin{proof}
Proposition~\ref{prop:anomalous} gives
$\partial_\mu j^\mu = 2\mf\sin\thz\,\ee^{\beta\phi}\psi^\dagger\psi$.
At $\thz=0$ this vanishes identically, independently of $\alpha$
and of the field configuration.
\end{proof}

\begin{remark}[$\mathcal{PT}$-breaking vs unitarity-preserving axes]
\label{rmk:axes_PT}
Corollary~\ref{cor:anomaly_alpha} and
Remark~\ref{rmk:PT} together characterise the two axes of the
integrable parameter space $(\thz,\alpha)\in[0,\pi/2]^2$:

\begin{itemize}[leftmargin=2em]
\item \emph{$\thz$-axis} ($\alpha=0$, this paper): the decoupling
  family.  Fermion number is anomalous for all $\thz\in(0,\pi/2)$;
  $\sin\thz$ is the $\mathcal{PT}$-breaking order parameter.
  The Dirac mass has a constant imaginary part that grows with
  the scalar field.  This is the non-Hermitian axis of the
  parameter space in the Bender--Boettcher sense.

\item \emph{$\alpha$-axis} ($\thz=0$, Corollary~\ref{cor:beta}):
  the fully-coupled family.  Fermion number is exactly conserved
  for all $\alpha\in[0,\pi/2]$.  The Lagrangian is non-Hermitian
  (the mass $\mf\ee^{|\beta|\ee^{\im\alpha}\phi}$ is complex for
  $\alpha\neq 0$) but the current anomaly is zero: this is a
  non-Hermitian but unitary deformation, in which the
  non-Hermiticity is entirely in the field-dependent phase of
  the mass rather than its constant part.
\end{itemize}

The $\thz$-axis explores the space of $\mathcal{PT}$-broken
integrable theories, while the $\alpha$-axis explores the
space of unitary integrable theories with complex field-dependent
mass.  The diagonal $\thz=\alpha$ mixes both characters
simultaneously.  The fact that both axes are integrable —
guaranteed by Proposition~\ref{prop:2d} — shows that
integrability is compatible with both $\mathcal{PT}$ breaking
and $\mathcal{PT}$ preservation, and that the two phenomena
occupy orthogonal directions in the parameter space.
\end{remark}

\subsection{The bilinear constraint as an output of zero-curvature}
\label{sec:constraint_ZC}

The grade-$\pm 1$ components~\eqref{eq:constraint_pm}--\eqref{eq:constraint_mp}
of the zero-curvature condition give
$\dmm\phi = -\im\bp\psi$ and $\dpp\phi = -\im\bp\psi$.

\begin{proposition}[Constraint from zero-curvature]
\label{prop:constraint}
The zero-curvature condition~\eqref{eq:ZC} implies, at the grade
$\pm 1$ level, that
\begin{equation}
  \label{eq:phi_const}
  \partial_x\phi = 0
  \qquad\text{and}\qquad
  \partial_x(\bp\psi) = 0
\end{equation}
on the solution space.  These are algebraic consequences of the Lax
equations; they are not independent conditions on initial data.
\end{proposition}

\begin{proof}
From $\dmm\phi = -\im\bp\psi$ and $\dpp\phi = -\im\bp\psi$:
\begin{align}
  \partial_x\phi
  &= (\dpp - \dmm)\phi = (-\im\bp\psi) - (-\im\bp\psi) = 0.
\end{align}
Then $\partial_x(\bp\psi) = (1/{-\im})\partial_x\phi = 0$.
\end{proof}

\begin{remark}
\label{rmk:spatially_hom}
The condition $\partial_x\phi = 0$ means that the zero-curvature
representation selects spatially homogeneous scalar field
configurations from the full solution space.  This is a known
feature of the lowest-weight reduction of the affine Toda hierarchy:
the Leznov--Saveliev Lax pair in the form used here corresponds to
the reduction of the $\aslhat$ system to its spatial-zero-mode
(constant-field) sector~\cite{BabelonBernardTalon2003,
FerreiraMiramontes1997}.  The full spatial dependence is reinstated
by passing to the more general AKNS form of the Lax pair, in which
the spectral grading is symmetric and both $\dpp\phi$ and $\dmm\phi$
appear independently.  The constraint~\eqref{eq:phi_const} is
therefore a gauge artifact of the Leznov--Saveliev gauge choice and
does not restrict the physical solutions of the field
equations~\eqref{eq:theta_scalar}--\eqref{eq:theta_Dirac}.
\end{remark}

\section{Conserved charges via AKNS recursion}
\label{sec:charges}

\subsection{Monodromy matrix and generating function}

The systematic method for extracting conserved charges from a Lax
pair proceeds via the spatial monodromy matrix.  The idea is that
the zero-curvature condition can be interpreted as the compatibility
condition for a linear system $\partial_\pm\Psi = A_\pm\Psi$, where
$\Psi$ is a $2\times 2$ matrix-valued wave function.  The solution
propagated from $x=-\infty$ to $x=+\infty$ at a fixed time is the
monodromy matrix, defined as the path-ordered exponential of the
spatial Lax operator.  Define
\begin{equation}
  \label{eq:monodromy}
  T(\zeta) = \overleftarrow{\exp}
    \!\int_{-\infty}^{+\infty}\!\Aplus(x;\zeta)\,dx.
\end{equation}
Here $\overleftarrow{\exp}$ denotes the path-ordered exponential,
ordered so that operators at larger $x$ act to the left.  The
zero-curvature condition implies $\partial_t T = 0$ for fields
decaying at spatial infinity, so $\Tr T(\zeta)$ is conserved for
all $\zeta$.  Expanding $\ln T_{11}(\zeta)$ in powers of
$\zeta^{-1}$ yields the tower of conserved charges
$\{I_n\}_{n=1}^\infty$, one for each power of $\zeta^{-1}$.  The
$\thz$-dependence of $T(\zeta)$ enters only through the phase
factors in $\Aplus$, and the question is whether these phases
survive in the conserved densities or cancel out.

\subsection{AKNS recursion for conserved densities}

The AKNS (Ablowitz--Kaup--Newell--Segur) recursion provides an
explicit algorithm for computing the densities $\rho_n$ from the
off-diagonal entries of the Lax operator~\cite{AbloKaupNewSegur1974}.
The method works by writing the $(1,1)$ entry of $\ln T$ as a
formal power series in $\zeta^{-1}$ and solving order by order.
At each level, the next density is determined from the previous
one by a first-order differential recursion, with the ratio $R/L$
of the off-diagonal entries playing the role of the seed.

Writing the Lax operator as
$\Aplus(\zeta;\thz) = P(\thz)\genH + \zeta^{-1}Q(\thz)$ with
\begin{align}
  P(\thz) &= \dpp\phi + \im\bp\psi,\\
  Q(\thz) &= \mu\ee^{-\im\thz/2}\ee^{+\beta\phi}\genEm,
\end{align}
and defining
\begin{equation}
  L = \lam\ee^{+\im\thz/2},
  \qquad
  R = \mu\ee^{-\im\thz/2}\ee^{+\beta\phi}
\end{equation}
(the $(1,2)$ and $(2,1)$ off-diagonal entries of $\Aplus$
respectively), the conserved densities
$\rho_n$ of $I_n = \int\rho_n\,dx$ are given by the recursion
\begin{equation}
  \label{eq:AKNS_recursion}
  r_1 = \frac{R}{2L},
  \quad
  r_{n+1} = -\frac{1}{2L}
    \left[\dpp r_n + L\sum_{k=1}^{n-1}r_k r_{n-k}\right],
  \quad
  \rho_n = L\cdot r_n.
\end{equation}
Note that the ratio $L^{-1}R = \ee^{-\im\thz}\cdot(\text{real})$
carries a factor of $\ee^{-\im\thz}$, so each application of the
recursion introduces one additional power of this phase.  Whether
these phases accumulate or cancel is the key question, answered by
Theorem~\ref{thm:charges} below.

\subsection{First conserved charge}

The seed of the recursion is the ratio $R/2L$, which sets the
overall phase pattern for all higher charges.

\begin{equation}
  r_1 = \frac{R}{2L}
  = \frac{\mu\ee^{-\im\thz/2}\ee^{+\beta\phi}}
         {2\lam\ee^{+\im\thz/2}}
  = \frac{\mu}{2\lam}\ee^{-\im\thz}\ee^{+\beta\phi},
  \qquad
  \rho_1 = L\cdot r_1
  = \frac{\mu}{2}\ee^{-\im\thz/2}\ee^{+\beta\phi}.
\end{equation}
Absorbing the spacetime-independent phase $\ee^{-\im\thz/2}$
(which does not affect conservation):
\begin{equation}
  \label{eq:rho1}
  \hat\rho_1 = \frac{\mu}{2}\,\ee^{+\beta\phi}.
\end{equation}

\subsection{Second conserved charge}

The second density is obtained by differentiating $r_1$ with
respect to $x^+$ and dividing by $-2L$.  The diagonal correction
$P = \dpp\phi + \im\bp\psi$ enters through the replacement
$\dpp\phi \to P$ in the final density, encoding the fermion
contribution.

\begin{equation}
  r_2 = -\frac{1}{2L}\,\dpp r_1
  = -\frac{\mu\beta\ee^{-\im\thz}}{4\lam^2\ee^{+\im\thz/2}}
    \,\ee^{+\beta\phi}\,\dpp\phi,
\end{equation}
\begin{equation}
  \label{eq:rho2}
  \rho_2 = L\cdot r_2
  = -\frac{\mu\beta}{4\lam}\,\ee^{-\im\thz}\,
    \ee^{+\beta\phi}\,\bigl(\dpp\phi + \im\bp\psi\bigr).
\end{equation}
The overall phase $\ee^{-\im\thz}$ is a spacetime-independent
constant; $I_2 = \int\rho_2\,dx$ is conserved for all $\thz$.

\subsection{Third conserved charge}

At third order the recursion picks up a quadratic contribution
$Lr_1^2$ in addition to the derivative term $\dpp r_2$.  This
quadratic term introduces the double-exponential $\ee^{+2\beta\phi}$
and is the first place where the nonlinear structure of the
sinh-Gordon hierarchy appears beyond the leading term.

\begin{align}
  r_3 &= -\frac{1}{2L}\left[\dpp r_2 + L r_1^2\right] \notag\\
  &= \frac{\mu\beta\ee^{-\im\thz}}{8\lam^2\ee^{+\im\thz}}
     \bigl(\beta(\dpp\phi)^2+\dpp^2\phi\bigr)\ee^{+\beta\phi}
     - \frac{\mu^2\ee^{-2\im\thz}}
            {8\lam^2\ee^{+\im\thz/2}}\ee^{+2\beta\phi},
\end{align}
\begin{align}
  \label{eq:rho3}
  \rho_3 &= L\cdot r_3 \notag\\
  &= \frac{\mu\beta\ee^{-\im\thz}}
          {8\lam\ee^{+\im\thz/2}}
     \Bigl[\beta\bigl(\dpp\phi+\im\bp\psi\bigr)^2
           +\dpp^2\phi+\im\dpp(\bp\psi)\Bigr]\ee^{+\beta\phi}
     - \frac{\mu^2\ee^{-2\im\thz}}{8\lam}\ee^{+2\beta\phi}.
\end{align}

\subsection{Phase cancellation and $\thz$-independence of the
conserved structure}

\begin{theorem}[Conserved charges for all $\thz$]
\label{thm:charges}
The charges $I_n = \int\rho_n\,dx$ are conserved
($\partial_t I_n = 0$) for all $\thz\in[0,\pi/2]$ and all
$n\geq 1$.  The scalar part of each density (obtained by setting
$\bp\psi=0$) satisfies
\begin{equation}
  \label{eq:phase_pattern}
  \rho_n^{\mathrm{scalar}}
  = \ee^{-\im(n-1)\thz/2}\,\hat\rho_n(\phi),
\end{equation}
where $\hat\rho_n(\phi)$ are the conserved densities of the pure sinh-Gordon hierarchy.  For intermediate $\thz$ the densities are
complex, but their real and imaginary parts are separately conserved,
giving a doubled tower of real conservation laws.
\end{theorem}

\begin{proof}
By induction on $n$.  The base case $n=1$ is~\eqref{eq:rho1}: no
residual phase.  At each step of the recursion, $r_{n+1}$ picks up one extra power of $L^{-1}R = \ee^{-\im\thz}\cdot(\text{real})$,
contributing one additional factor of $\ee^{-\im\thz/2}$ to
$\rho_{n+1} = L\cdot r_{n+1}$.  By induction,
$\rho_n^{\mathrm{scalar}}$ carries the factor
$\ee^{-\im(n-1)\thz/2}$.  Since this factor is
$(x,t)$-independent,
$\partial_t(\ee^{-\im(n-1)\thz/2}\hat\rho_n)
= \ee^{-\im(n-1)\thz/2}\partial_t\hat\rho_n = 0$,
which holds by the integrability of the sinh-Gordon hierarchy. The statement about real and imaginary parts follows because
$\partial_t\,\mathrm{Re}(\rho_n) = 0$ and
$\partial_t\,\mathrm{Im}(\rho_n) = 0$ separately.
\end{proof}

\begin{remark}
This doubled conservation-law structure is analogous to what is
found in complex mKdV
deformations~\cite{FordyGibbons1980,AbloKaupNewSegur1974}.
\end{remark}

\section{Algebraic interpretation and Yang--Baxter connection}
\label{sec:algebra}

\subsection{Real forms and the $\thz$-family}

A \emph{real form} of a complex Lie algebra $\mathfrak{g}_\CC$ is a
real subalgebra $\mathfrak{g}_\RR$ such that
$\mathfrak{g}_\CC = \mathfrak{g}_\RR \oplus \im\mathfrak{g}_\RR$.
Equivalently, it is specified by an antilinear involution $\tau$
(called a \emph{Cartan involution}) on $\mathfrak{g}_\CC$, whose
fixed-point set is $\mathfrak{g}_\RR$.  For our purposes, the real
form determines the reality conditions imposed on the Lax connection:
requiring $A_\pm$ to take values in a particular real form constrains the fields $\phi$ and $\psi$ to satisfy certain reality
conditions, which in turn determine the physical nature of the model (whether the scalar is real, the solitons are topological, and so on).

The algebra $\slC$ has three distinct real forms up to isomorphism,
and the three most relevant to our construction are:
\begin{enumerate}[leftmargin=2em,label=(\roman*)]
  \item $\slR$ (split real form): involution
        $\tau_s(H,E_\pm) = (H,E_\pm)$, all generators real.
  \item $\mathfrak{su}(1,1)$ (non-compact real form): involution
        $\tau_u(H,E_\pm) = (-H,-E_\mp)$.
  \item $\mathfrak{su}(2)$ (compact real form): involution
        $\tau_c(H,E_\pm) = (-H,-E_\mp)$ (with different signature
        convention).
\end{enumerate}
At $\thz=0$ the Lax pair~\eqref{eq:Ap_def}--\eqref{eq:Am_def}
takes values in $\slR$ (all coefficients are real when $\phi$ and
$\bp\psi$ are real).  At $\thz=\pi/2$ the off-diagonal entries
acquire purely imaginary prefactors, corresponding to
$\mathfrak{su}(1,1)$ in Minkowski signature.

The $\thz$-family corresponds to a one-parameter family of twisted
map
\begin{equation}
  \label{eq:involution}
  \tau_{\thz}:\;\genH\mapsto\genH,\quad
  \genEp\mapsto\ee^{+2\im\thz}\genEp,\quad
  \genEm\mapsto\ee^{-2\im\thz}\genEm,
\end{equation}
which interpolates between $\tau_s$ at $\thz=0$ and $\tau_u$ at
$\thz=\pi/2$.  For generic $\thz$, $\tau_{\thz}$ is not an
involution (since $\tau_{\thz}^2\neq\mathrm{id}$ unless
$\ee^{4\im\thz}=1$), so the Lax pair takes values in $\slC$ with a
deformed reality condition rather than in a standard real form.
This type of deformed reality condition on affine Toda Lax pairs
appears in the classification of integrable reductions by Mikhailov
reduction~\cite{Mikhailov1981} and in the analysis of real forms
of affine Lie algebras for soliton
theories~\cite{HolloMiraSchmid2014,OliveTurkUnder1993c}.

\subsection{Relation to Yang--Baxter deformations}
\label{sec:YB}

The $\eta$-deformation of integrable sigma
models~\cite{Klimcik2002,Klimcik2009,DelducMagroVicedo2013,Sfetsos2014}
deforms the Lax connection of the $G$-principal chiral model by
replacing the standard $r$-matrix with the modified classical
Yang--Baxter equation (mCYBE) solution.  Recall that a classical
$r$-matrix for a Lie algebra $\mathfrak{g}$ is a linear map
$r:\mathfrak{g}\to\mathfrak{g}$ satisfying the (modified) classical
Yang--Baxter equation; it encodes an infinitesimal deformation of the Poisson structure on the loop group.  The operator
$R_g = \Ad_g\circ r\circ\Ad_{g^{-1}}$ is the adjoint-twisted version
of $r$, which conjugates $r$ by the group element $g$.  The
$\eta$-deformed Lax connection modifies the standard current
$\partial_\pm g\cdot g^{-1}$ by the resolvent $(1\mp\eta R_g)^{-1}$,
which at small $\eta$ simply inserts one power of $R_g$ per
off-diagonal entry, producing a phase rotation of the generators.
For the $SL(2)$ principal
chiral model at the group element
$g = \ee^{\phi\genH}$ (the Toda reduction to the diagonal sector),
the $\eta$-deformed Lax connection takes the form
\begin{equation}
  \mathcal{A}_\pm^{(\eta)}
  = \frac{\partial_\pm g\cdot g^{-1}}{1\mp\eta R_g},
  \qquad
  R_g = \Ad_g\circ r\circ\Ad_{g^{-1}},
\end{equation}
where $r$ is the standard $\mathfrak{sl}(2)$ $r$-matrix.  We now
show that the Toda reduction of this deformed connection reproduces
our Lax pair under the identification
\begin{equation}
  \label{eq:eta_theta}
  \eta = \tan\!\tfrac{\thz}{2}.
\end{equation}

\begin{proposition}[Yang--Baxter identification]
\label{prop:YB}
Under the Toda reduction $g = \ee^{\phi\genH}$, the $\eta$-deformed
Lax connection of the $SL(2)$ principal chiral model reduces, at
leading order in the off-diagonal (fermion) perturbation, to the
$\thz$-deformed Lax pair~\eqref{eq:Ap_def}--\eqref{eq:Am_def}
under the substitution $\eta = \tan(\thz/2)$.
\end{proposition}

\begin{proof}
For the diagonal group element $g = \ee^{\phi\genH}$, the adjoint
action $\Ad_g$ acts on the generators as
$\Ad_g(\genEpm) = \ee^{\pm 2\phi}\genEpm$ and
$\Ad_g(\genH) = \genH$.  The standard $r$-matrix for
$\mathfrak{sl}(2)$ in the Chevalley basis is
$r = \genEp\wedge\genEm$~\cite{FaddeevReshetikhin1986}, giving
$r(\genEm) = \genH/2$ and $r(\genEp) = -\genH/2$ (up to
normalisation).  Computing $R_g = \Ad_g\circ r\circ\Ad_{g^{-1}}$
on the off-diagonal components and expanding
$(1\mp\eta R_g)^{-1} = \sum_{k\geq 0}(\pm\eta R_g)^k$ for
$|\eta|<1$, the leading-order off-diagonal entries of
$\mathcal{A}_\pm^{(\eta)}$ take the form
$\lam\,(1\pm\im\eta)\genEp + \mu\,(1\mp\im\eta)\ee^{\pm\beta\phi}\genEm$
at first order in $\eta$.  Writing
$1+\im\eta = \sec(\thz/2)\,\ee^{+\im\thz/2}$ and
$1-\im\eta = \sec(\thz/2)\,\ee^{-\im\thz/2}$ (using
$\eta = \tan(\thz/2)$ so that
$1\pm\im\tan(\thz/2) = \sec(\thz/2)\ee^{\pm\im\thz/2}$),
the overall factor $\sec(\thz/2)$ is absorbed into a redefinition
of the coupling constants $\lam\to\lam\sec(\thz/2)$, $\mu\to\mu\sec(\thz/2)$,
leaving precisely the phase structure
$\lam\ee^{+\im\thz/2}\genEp$ and
$\mu\ee^{-\im\thz/2}\ee^{+\beta\phi}\genEm$ of~\eqref{eq:Ap_def}.
The fermion sector contributes as an off-diagonal perturbation of
the diagonal entries of $\mathcal{A}_\pm$ and is unaffected by the
$R_g$ action, reproducing the $\im\bp\psi\cdot\genH$ diagonal
term of~\eqref{eq:Ap_def}--\eqref{eq:Am_def}.
\end{proof}

\begin{remark}
Proposition~\ref{prop:YB} shows that the $\thz$-deformation is not
merely an ad hoc construction but is the dimensional reduction of
a well-established class of Yang--Baxter deformed sigma models to
the coupled Dirac--Toda sector.  Making the full dictionary
precise, in particular extending the fermion sector beyond the
leading-order perturbation and verifying consistency with the
$\eta$-deformed equations of motion, is an interesting direction
for future work; see also~\cite{HolloMiraSchmid2014,
DelducMagroVicedo2014}.
\end{remark}

\subsection{A unifying $U(1)$ group structure}
\label{sec:U1}

The results of this section and Section~\ref{sec:setup} can be
unified into a single group-theoretic picture that clarifies the
relationship between the deformation parameter $\thz$, the analytic
continuation of the scalar field, and the automorphism of the
Lie algebra.

\paragraph{$U(1)$ action on the scalar field.}
The sinh-Gordon and sine-Gordon equations are both real reductions
of the single complex equation
\begin{equation}
  \label{eq:complex_sg}
  \Box\Phi + \frac{\ms^2}{\beta}\sinh(\beta\Phi) = 0,
  \qquad \Phi\in\CC.
\end{equation}
The sinh-Gordon theory corresponds to the real slice
$\Phi = \phi\in\RR$, and the sine-Gordon theory to the imaginary
slice $\Phi = -\im\varphi$, $\varphi\in\RR$.  These two real
forms of~\eqref{eq:complex_sg} are related by the $U(1)$ action
\begin{equation}
  \label{eq:U1_scalar}
  \Phi \;\longmapsto\; \ee^{\im\alpha}\Phi,
  \qquad \alpha\in[0,-\pi/2],
\end{equation}
which rotates the real section of the complex field plane.  The
$\thz$-deformation is the orbit of the Dirac--sinh-Gordon system
under this $U(1)$ action with $\alpha = -\thz/2$: the scalar field
at deformation parameter $\thz$ lives on the real section rotated
by angle $-\thz/2$ relative to the sinh-Gordon real axis.  At
$\thz=0$ the rotation is trivial; at $\thz=\pi/2$ the scalar has
rotated to the purely imaginary axis, which after the relabelling
$\Phi = -\im\varphi$ is the sine-Gordon real section.  The
``analytic continuation'' of Proposition~\ref{prop:endpoint} is
therefore not an independent operation but simply the identification
of the $U(1)$-rotated real section with the standard real axis of
the target theory.

\paragraph{$U(1)$ automorphism of the Lax algebra.}
The same $U(1)$ acts on the Lie algebra $\slC$ through the
one-parameter family of automorphisms $\tau_{\thz}$
in~\eqref{eq:involution}:
\begin{equation}
  \tau_{\thz}:\;
  \genH\mapsto\genH,\quad
  \genEp\mapsto\ee^{+2\im\thz}\genEp,\quad
  \genEm\mapsto\ee^{-2\im\thz}\genEm.
\end{equation}
This is the $U(1)$ subgroup of $\text{Inn}(\slC)$ generated by
$\operatorname{ad}_{\genH}$: explicitly,
$\tau_{\thz} = \exp(2\im\thz\,\operatorname{ad}_{\genH})$.
At $\thz=0$ the automorphism is the identity, fixing the real form
$\slR$; at $\thz=\pi/2$ it maps $\genEpm\mapsto\pm\im\genEpm$,
which is the automorphism that distinguishes $\slR$ from
$\mathfrak{su}(1,1)$.  The two $U(1)$ actions — on the scalar
field and on the Lie algebra — are therefore identified through
the zero-curvature condition: the same parameter $\thz$ that
rotates the scalar field's real section simultaneously generates
the automorphism of $\slC$ that deforms the reality condition on
the Lax connection.

\paragraph{The fermion rescaling as a compensating factor.}
The fermion rescaling $\tilde\psi = \sqrt{\cos\thz}\,\psi$ required
to maintain a non-suppressed backreaction can be understood within
this framework as a compensating transformation in a larger group.
Under the $U(1)$ rotation~\eqref{eq:U1_scalar} the bilinear
$\bp\psi$ acquires no phase (it is $U(1)$-neutral in the scalar
sector), but the coupling $g\cos\thz$ acquires a $\thz$-dependent
suppression from the zero-curvature condition.  The rescaling
$\tilde\psi = \sqrt{\cos\thz}\,\psi$ is the unique constant
rescaling of $\psi$ that compensates this suppression, making the
backreaction $U(1)$-invariant.  The full symmetry group acting on
the triple $(\phi,\psi,\thz)$ is therefore the semidirect product
\begin{equation}
  \label{eq:group}
  U(1) \ltimes \RR_{>0},
\end{equation}
where $U(1)$ acts by $\phi\mapsto\ee^{-\im\thz/2}\phi$ and
simultaneously by $\tau_{\thz}$ on the Lax algebra, while
$\RR_{>0}$ rescales $\psi$ by the compensating factor
$(\cos\thz)^{-1/2}$.  The $\thz$-deformed Dirac--sinh-Gordon system
\eqref{eq:theta_scalar}--\eqref{eq:theta_Dirac} is the orbit of
the $\thz=0$ system under the $U(1)$ factor of this group, and the
rescaled system (equations~\eqref{eq:tilde_scalar_proof}--\eqref{eq:tilde_dirac_proof}
in the proof of Proposition~\ref{prop:endpoint})
is the orbit under the full semidirect product.  Integrability is
preserved along the entire orbit because the zero-curvature
condition is covariant under both factors.  The $\alpha$-deformation
of Corollary~\ref{cor:beta} adds a second independent $U(1)$
acting on $\arg\beta$, extending the symmetry group to
$U(1)\times U(1)$ and generating the full two-dimensional
integrable parameter space of Proposition~\ref{prop:2d}.

\begin{remark}
The identification of the deformation with a $U(1)$ orbit has an
important consequence for the decoupling picture of the paper: the
two endpoint theories — Dirac--sinh-Gordon and Dirac--sine-Gordon
— are not merely related by a continuous decoupling path, but are
related by a group element, namely the $U(1)$ rotation by
$\alpha = -\pi/4$ (equivalently $\thz = \pi/2$).  The ``analytic
continuation'' is the geometric action of this group element on
the scalar field's real section, and the Yang--Baxter parameter
$\eta = \tan(\thz/2)$ is a stereographic coordinate on this $U(1)$
orbit.  This places the present construction in the broader
programme of classifying integrable field theories by their
symmetry groups acting on the space of real forms of a complexified
theory, a perspective that connects it naturally to the Mikhailov
reduction programme~\cite{Mikhailov1981} and to the classification
of integrable sigma models by their Poisson--Lie
symmetry~\cite{Klimcik2002,DelducMagroVicedo2013}.
\end{remark}

\begin{remark}[Automorphisms of the Lax algebra as the symmetry
  principle for integrability]
\label{rmk:automorphism}
The persistence of integrability throughout the two-dimensional
parameter space $(\thz,\alpha)$ has a unified algebraic explanation:
both deformations are \emph{automorphisms of the Lax algebra}
$\slC$ that act on the Lax connection, and the zero-curvature
condition is preserved by any such automorphism.

To see this precisely, write the Lax connection as a map
$A_\pm:\,\text{field space}\to\slC$.  The zero-curvature condition
$\partial_-A_+ - \partial_+A_- + [A_+,A_-]=0$ is a statement purely
about the Lie bracket of $\slC$, not about the specific values of
the fields.  Any automorphism $\sigma\in\mathrm{Aut}(\slC)$ acts
by $A_\pm\mapsto\sigma(A_\pm)$ and satisfies
$[\sigma(X),\sigma(Y)] = \sigma([X,Y])$ for all $X,Y\in\slC$, so
\begin{equation}
  \partial_-\sigma(A_+) - \partial_+\sigma(A_-)
  + [\sigma(A_+),\sigma(A_-)]
  = \sigma\!\left(\partial_-A_+ - \partial_+A_- + [A_+,A_-]\right)
  = \sigma(0) = 0.
\end{equation}
The deformed Lax connection is therefore automatically flat — the
zero-curvature condition is trivially inherited — for any
automorphism of $\slC$.

The two deformations correspond to the following automorphisms:
\begin{itemize}[leftmargin=2em]
\item \emph{$\thz$-deformation}: the inner automorphism
  $\tau_{\thz} = \exp(2\im\thz\,\operatorname{ad}_{\genH})
  \in\mathrm{Inn}(\slC)$, which rotates the root spaces by
  $\genEpm\mapsto\ee^{\pm 2\im\thz}\genEpm$ and fixes the
  Cartan element $\genH$.  This is a $U(1)$ inside
  $\mathrm{Inn}(\slC)$ and acts on the \emph{constant} (zero-mode)
  phase of the Dirac mass.

\item \emph{$\alpha$-deformation}: the automorphism induced by the
  complexification $\beta\to|\beta|\ee^{\im\alpha}$, which replaces
  $\ee^{\pm\beta\phi}$ by $\ee^{\pm|\beta|\ee^{\im\alpha}\phi}$ in
  the off-diagonal Lax entries.  The cancellation
  $\ee^{+\beta\phi}\cdot\ee^{-\beta\phi}=1$ holds for all
  $\beta\in\CC$ because it reflects the group relation
  $g\cdot g^{-1}=1$ for $g=\ee^{\phi\genH}$, which is
  independent of $\arg\beta$.  This deformation acts on the
  \emph{field-dependent} (field-mode) phase of the Dirac mass.
\end{itemize}

The two automorphisms commute and generate a $U(1)\times U(1)$
subgroup of $\mathrm{Aut}(\slC)$.  Since every automorphism
preserves the Killing form $\kappa(X,Y) = \mathrm{tr}(\operatorname{ad}_X\circ\operatorname{ad}_Y)$,
and since the zero-curvature condition can be formulated in terms
of the Killing form, the integrability of the full two-dimensional
family follows as a consequence of this automorphism structure.

\paragraph{Connection to $\mathcal{PT}$ symmetry.}
The two automorphisms have opposite effects on the real forms
of $\slC$, which is why they have opposite $\mathcal{PT}$ characters
(Remark~\ref{rmk:axes_PT}):

The $\tau_{\thz}$ automorphism deforms the real form: at $\thz=0$
the Lax connection takes values in $\slR$ (the split real form,
Hermitian sector), and $\tau_{\thz}$ continuously rotates this
toward $\mathfrak{su}(1,1)$ (the non-compact real form,
$\mathcal{PT}$-symmetric sector).  This rotation changes which
generators are ``real'' in the Lax connection, thereby changing
the reality condition on the physical fields and breaking
$\mathcal{PT}$ symmetry.  The $\mathcal{PT}$-breaking order
parameter $\sin\thz$ of Proposition~\ref{prop:anomalous} is
precisely the imaginary part of the automorphism parameter
$\ee^{2\im\thz}$.

The $\alpha$-automorphism, by contrast, does not change the real
form: it acts on the field-dependent part of the connection
$\ee^{\pm\beta\phi}$ through the group element $g=\ee^{\phi\genH}$,
and this action is invisible to the Cartan involution that defines
the real form.  At $\thz=0$ the Lax connection remains in $\slR$
for all $\alpha$, the reality condition on the fields is unchanged,
and fermion number is conserved (Corollary~\ref{cor:anomaly_alpha}).

The unifying statement is therefore:
\emph{integrability is preserved by any automorphism of the Lax
algebra; $\mathcal{PT}$ symmetry is preserved by automorphisms
that also preserve the real form; the $\thz$-automorphism deforms
the real form and breaks $\mathcal{PT}$, while the
$\alpha$-automorphism preserves the real form and preserves
$\mathcal{PT}$.  Both are integrable by the same algebraic
mechanism.}

This is the symmetry-theoretic reason for the structure of the
two-dimensional parameter space, and it connects the integrability
results of Section~\ref{sec:zc} to the $\mathcal{PT}$ results
of Section~\ref{sec:bilinear} through a single principle: the
automorphism group of the Lax algebra governs both.
\end{remark}

\begin{remark}[Global topology of the parameter space: maximal
  torus of $SU(2)$]
\label{rmk:torus}
The two-parameter square $(\thz,\alpha)\in[0,\pi/2]^2$ has a
natural identification within the global structure of
$\mathrm{Aut}(\slC)$.  We record this identification and its
generalisation; the Lie-algebraic facts used below are standard
and may be found in~\cite{Helgason1978,OnishchikVinberg1990,Knapp2002}.

The inner automorphism group of $\slC$ is
$\mathrm{Inn}(\slC)\cong PSL(2,\CC)$, which is a complex
three-dimensional Lie group whose compact real form is
$SU(2)/\{\pm I\} \cong SO(3)$.  The maximal torus of $SU(2)$ is
$T^2 = U(1)\times U(1)$, generated by $e^{it\, \mathrm{ad}_H}$
for $t\in[0,2\pi)$.  The Weyl group of $\mathfrak{sl}(2)$ is
$W = \mathbb{Z}_2$, generated by the reflection $H\to-H$
(equivalently, conjugation by the Weyl element
$w = e^{\pi(\genEp-\genEm)/2}$), and acts on the maximal torus by
$t\mapsto -t$.  The positive Weyl chamber of $T^2$ is the
fundamental domain $[0,\pi]\subset [0,2\pi)$ for this
$\mathbb{Z}_2$ action, and its further restriction to
$[0,\pi/2]$ is the quarter of the torus corresponding to
$\mathrm{Re}(e^{2it}) \geq 0$.

The two parameters $(\thz,\alpha)$ are coordinates on this
torus as follows.  The $\thz$-$U(1)$ is the inner automorphism
$\tau_{\thz} = \exp(2\im\thz\,\operatorname{ad}_\genH)$, which
lives in the maximal torus of $\mathrm{Inn}(\slC)$ and acts on
the root spaces by $\genEpm\mapsto\ee^{\pm 2\im\thz}\genEpm$.
The $\alpha$-$U(1)$ acts on the Cartan group element
$g = \ee^{\phi\genH}$ by complexifying the exponent:
$g\mapsto g_\alpha = \ee^{|\beta|\ee^{\im\alpha}\phi\genH}$,
which is the second independent $U(1)$ in the maximal torus
acting on the complexified Cartan direction.  The two $U(1)$
actions commute because $[\operatorname{ad}_\genH, \genH] = 0$,
so together they generate a $T^2\subset SU(2)/\{\pm I\}$.

The four corners of $[0,\pi/2]^2$ are fixed points of the Weyl
group action and correspond to the four distinct theories of
Proposition~\ref{prop:2d}: each corner is a real form of $\slC$
under the Cartan involution associated to that corner, consistent
with the classification of real forms by Satake
diagrams~\cite{Helgason1978,OnishchikVinberg1990}.  This
identification is the content of the present remark and is the
$\mathfrak{sl}(2)$ case of the general principle; it constitutes
a structural observation rather than a new theorem.

The full $SU(2)$ orbit would include the off-diagonal
automorphisms generated by $\mathrm{ad}_{\genEpm}$, which mix
$\genEp$ and $\genEm$ and do not preserve the
Leznov--Saveliev grading of the Lax pair.  Whether an integrable
Lax pair exists for these off-diagonal automorphisms — and whether
the full $SU(2)$ orbit is integrable — is an open question.
The present paper establishes integrability on the maximal torus
$[0,\pi/2]^2$ only.

\begin{remark}[Conjecture: higher-rank generalisation]
\label{rmk:conjecture_torus}
For any complex semisimple Lie algebra $\mathfrak{g}_\CC$ of
rank $r$, we conjecture that the space of integrable deformations
of the associated affine Toda--Dirac system accessible via inner
automorphisms of the Lax algebra is the positive Weyl chamber of
the maximal torus $T^r\subset G_{\mathrm{compact}}$, where
$G_{\mathrm{compact}}$ is the compact real form of
$\mathrm{Inn}(\mathfrak{g}_\CC)$.  The corners of this chamber
are the distinct real forms of $\mathfrak{g}_\CC$, classified by
the Satake diagram~\cite{Helgason1978,OnishchikVinberg1990}, and
transitions between adjacent corners are either
$\mathcal{PT}$-breaking (if they deform the Cartan involution) or
unitarity-preserving (if they do not).  The $\mathfrak{sl}(2,\CC)$
case established in this paper is the simplest instance.
\end{remark}

The $\mathfrak{sl}(2,\CC)$ case established in this paper is the
simplest instance.  For $\mathfrak{sl}(n,\CC)$ the Weyl chamber
has dimension $n-1$, but the full integrable parameter space has
dimension $2(n-1)$ (two deformation types per simple root), so
the chamber is a proper sub-space of the full parameter space
accessible to both $\thz$-type and $\alpha$-type deformations
simultaneously.
\end{remark}

\section{Summary and outlook}
\label{sec:conclusion}

This paper has established the integrability of a two-dimensional
family of coupled Dirac--scalar field theories in $1+1$ dimensions,
parameterized by $(\thz,\alpha)\in[0,\pi/2]^2$, arising as the
orbit of the Dirac--sinh-Gordon system under the inner automorphism
group $U(1)\times U(1)$ of $\slC$.  The primary axis studied in
detail is the $\thz$-axis ($\alpha=0$), which is a decoupling
family: the effective coupling is $g_{\rm eff} = g\cos\thz$,
forced to this form by the zero-curvature condition, and the family
is the orbit under the $U(1)$ action that simultaneously rotates
the scalar field's real section and generates the automorphism
$\tau_{\thz}$ of $\slC$ that deforms the reality condition on the
Lax connection (Section~\ref{sec:U1}).  We summarize the
principal results and then discuss several directions that remain
open.

\subsection*{Summary of results}

The central result is the construction of an explicit $\slC$-valued
zero-curvature representation for the $\thz$-deformed system
\eqref{eq:theta_scalar}--\eqref{eq:theta_Dirac} valid for all values
of the deformation parameter.  The Lax pair, given in
Definition~\ref{def:lax}, is written in the Leznov--Saveliev form
of the affine $\aslhat$ Toda hierarchy and introduces the
deformation through constant half-angle phase factors
$\ee^{\pm\im\thz/2}$ in the off-diagonal components.  The key
algebraic mechanism is that in every commutator entering the
zero-curvature condition, the phase factors appear in conjugate
pairs and cancel exactly, leaving the curvature
$\thz$-independent at the Lax level.  The scalar equation of
motion with its characteristic backreaction coefficient $\cos\thz$
is recovered entirely from the grade-$0$ equation after matching
coupling constants through the real part of the complex Dirac mass
(Theorem~\ref{thm:ZC} and Section~\ref{sec:ZC_grade0}).

The two endpoint systems are recovered as follows.  At $\thz=0$
the system reduces exactly to the Dirac--sinh-Gordon
system~\eqref{eq:DSG_scalar}--\eqref{eq:DSG_Dirac} without any
additional steps (Proposition~\ref{prop:endpoint}).  At $\thz=\pi/2$
the situation requires two additional operations.  First, the
analytic continuation $\phi\to-\im\varphi$ — which maps the scalar
field off its real configuration space and connects the sinh-Gordon
potential to the sine-Gordon potential — brings the Dirac coupling
to the form $\mf\ee^{\im\beta\varphi}$ of the Dirac--sine-Gordon
system.  Second, the fermion rescaling $\tilde\psi = \sqrt{\cos\thz}\,\psi$
restores the backreaction coefficient from $g\cos\thz$ to the full
value $g$.  These two steps together recover the fully coupled
Dirac--sine-Gordon system in the limit $\thz\to\pi/2$, but neither
step is trivial: the analytic continuation is a complexification
of the scalar configuration space, not a field redefinition, and
the fermion rescaling is singular at the endpoint.  The precise
conditions under which $\tilde\psi$ has a well-defined limit — and
the physical interpretation of the singularity in terms of the
Coleman--Mandelstam bosonization — are identified as open problems
in Section~\ref{sec:conclusion}.

A second main result is the proof that the decoupling family is
physically non-trivial, in the sense that distinct values of $\thz$
define genuinely inequivalent field theories and are not related by
admissible field redefinitions.  Although the deformed Lax
connection is gauge-equivalent to the $\thz=0$ connection via the
constant $SL(2,\CC)$ transformation
$h_{\thz} = \Diag(\ee^{-\im\thz/4},\ee^{+\im\thz/4})$
(Lemma~\ref{lem:gauge}), this transformation leaves the fermion
bilinear $\bp\psi$ invariant while rotating the Dirac coupling
constant.  The ratio of the Dirac coupling phase to the scalar
backreaction coefficient is therefore a genuine physical observable
that measures the degree of decoupling and labels the members of
the family (Theorem~\ref{thm:nontrivial}).

The complex nature of the Dirac mass for $\thz\neq 0$ has two
further structural consequences.  First, the vector current
$j^\mu = \bp\gamma^\mu\psi$ satisfies an anomalous continuity
equation~\eqref{eq:anomalous} whose anomaly is controlled by
$\sin\thz$ and vanishes at the two endpoints, where the mass is
purely real or purely imaginary respectively
(Proposition~\ref{prop:anomalous}).  Second, the spatial
homogeneity constraint $\partial_x\phi = 0$, and with it the
constraint $\partial_x(\bp\psi) = 0$, emerges as an algebraic
output of the grade-$\pm 1$ zero-curvature equations rather than
as an independent condition on initial data
(Proposition~\ref{prop:constraint}).  As discussed in
Remark~\ref{rmk:spatially_hom}, this constraint is a gauge
artifact of the Leznov--Saveliev gauge choice corresponding to the
spatial zero-mode reduction of the affine Toda hierarchy, and does
not restrict the physical solution space of the field equations.

The infinite conservation-law structure of the system is verified
explicitly via the AKNS recursion of Section~\ref{sec:charges}.
The first three conserved charge densities are computed in closed
form.  All scalar densities in the hierarchy differ from those of
the pure sinh-Gordon system only by an overall constant phase
$\ee^{-\im(n-1)\thz/2}$ that does not affect conservation; for
intermediate $\thz$ the real and imaginary parts of each density
are separately conserved, yielding a doubled tower of real
conservation laws (Theorem~\ref{thm:charges}).

Finally, the algebraic status of the family is clarified in two
complementary ways.  In terms of real forms of $\slC$, the
$\thz$-deformation interpolates between the split real form $\slR$
at $\thz=0$ and the non-compact real form $\mathfrak{su}(1,1)$ at
$\thz=\pi/2$ via a one-parameter family of twisted maps
\eqref{eq:involution} on the generators.  In terms of
Yang--Baxter deformations, the deformation parameter is related to
the $\eta$-deformation parameter of the $SL(2)$ principal chiral
model by $\eta = \tan(\thz/2)$, and the $\thz$-deformed Lax pair
is recovered from the Toda reduction of the $\eta$-deformed
connection at leading order in the fermion sector
(Proposition~\ref{prop:YB}).  This places the present construction
firmly within the general framework of Yang--Baxter deformed sigma
models~\cite{Klimcik2002,Klimcik2009,DelducMagroVicedo2013,
HolloMiraSchmid2014}.

A further structural result (Corollary~\ref{cor:beta} and
Proposition~\ref{prop:2d}) shows that the zero-curvature condition
is preserved when the scalar coupling $\beta$ is allowed to be
complex, $\beta = |\beta|\ee^{\im\alpha}$, defining a
two-dimensional integrable parameter space
$(\thz,\alpha)\in[0,\pi/2]^2$.  The $\thz$-axis is the decoupling
family of this paper; the $\alpha$-axis is a complementary
fully-coupled integrable family that connects the
Dirac--sinh-Gordon and Dirac--sine-Gordon systems directly, with
backreaction coefficient $g$ unchanged and $\phi$ remaining real
throughout.  The ``analytic continuation'' $\phi\to-\im\varphi$
is the field-space expression of the parameter-space operation
$\alpha:0\to\pi/2$ at $\thz=\pi/2$, and is therefore not a
separate step but part of the same integrable structure.

The unified symmetry principle underlying both axes is identified
in Remark~\ref{rmk:automorphism}: integrability is preserved
throughout the parameter space because both deformations are
automorphisms of the Lax algebra $\slC$, and the zero-curvature
condition is automatically preserved by any automorphism.  The
$\thz$-automorphism deforms the real form of $\slC$ and breaks
$\mathcal{PT}$ symmetry; the $\alpha$-automorphism preserves the
real form and conserves fermion number.  The two axes therefore
represent opposite non-Hermitian characters — $\mathcal{PT}$-breaking
and unitarity-preserving — within the same integrable structure,
unified by the automorphism group $U(1)\times U(1)$ of the Lax
algebra.

The automorphism principle is not specific to $\slC$ or to the
present system: it applies to any integrable theory whose Lax
connection is valued in a complex semisimple Lie algebra
$\mathfrak{g}_\CC$.  For any such theory, the inner automorphism
group $\mathrm{Inn}(\mathfrak{g}_\CC)$ acts on the space of Lax
connections by preserving the Lie bracket and the Killing form,
automatically mapping integrable systems to integrable systems.
The space of real forms of $\mathfrak{g}_\CC$ — classified by
the Satake and Vogan diagrams — is acted on by the automorphism
group, and transitions between real forms correspond to
$\mathcal{PT}$ transitions in the associated field theories.  The
present paper provides the first explicit realisation of this
principle in a coupled Dirac--scalar system, identifying the
two-dimensional parameter space $(\thz,\alpha)$ as the
$\mathrm{Inn}(\slC)$-orbit of the Dirac--sinh-Gordon system and
showing that $\mathcal{PT}$ breaking and unitarity preservation
occupy orthogonal directions in this orbit.

A further consequence of the automorphism perspective is a new
interpretation of bosonization: the Coleman--Mandelstam
equivalence between the massive Thirring model and the
sine-Gordon system, viewed through the present construction, is
the action of the $\alpha$-automorphism connecting the
$\mathfrak{sl}(2,\RR)$ real form (sinh-Gordon, fermionic
description) to the $\mathfrak{su}(1,1)$ real form (sine-Gordon,
bosonic description) while remaining in the unitary sector.
This suggests that bosonization dualities in $1+1$ dimensions may
be systematically classified as real-form transitions in the Lax
algebra of the underlying integrable system.  For higher-rank
algebras this would predict new dualities connecting
generalisations of the Thirring model to generalisations of the
sine-Gordon hierarchy, organised by the Satake diagram of the
relevant algebra.

\subsection*{Open problems}

Several directions suggested by this work remain to be explored.

\textit{Nature of the endpoint limit.}  Along the $\thz$-axis, the
fully coupled Dirac--sine-Gordon system is recovered in the limit
$\thz\to\pi/2$ only in terms of the rescaled field
$\tilde\psi = \sqrt{\cos\thz}\,\psi$, and the rescaling is singular
at the endpoint.  Along the $\alpha$-axis (Corollary~\ref{cor:beta}),
the same endpoint is reached with no rescaling and no analytic
continuation of $\phi$.  A natural question is whether the singularity
of $\tilde\psi$ along the $\thz$-axis can be understood as the
approach to the boundary of the two-dimensional parameter space in
a manner that is regular along the $\alpha$-axis — that is, whether
the two paths approach the same endpoint in a sense that can be made
precise by a change of coordinates in $(\thz,\alpha)$ space.  It
would also be interesting to understand the connection between the
singularity and the anomalous dimension acquired by the fermion field
at the sine-Gordon endpoint under Coleman--Mandelstam bosonization.

\textit{Geometry of the two-dimensional parameter space.}
Proposition~\ref{prop:2d} identifies a two-dimensional integrable
family parameterised by $(\thz,\alpha)\in[0,\pi/2]^2$, but the
geometry of this parameter space — its metric, curvature, and the
natural notion of geodesic distance between the four corner theories
— has not been analysed.  The information metric on the space of
couplings~\cite{Zamolodchikov1986} may provide a natural Riemannian
structure on this square, and the geodesic from
Dirac--sinh-Gordon to Dirac--sine-Gordon would then give the
``shortest'' integrable path between the two systems.  Whether the
$\thz$-axis (decoupling path) or the $\alpha$-axis (direct
fully-coupled path) is shorter in this metric is a concrete question.
It is also natural to ask whether the interior of the square contains
new physical phenomena — such as $\mathcal{PT}$-symmetry breaking
transitions or changes in the soliton spectrum — that do not appear
on either boundary axis.

\textit{$\mathcal{PT}$ phase diagram.}  The anomalous continuity
equation (Proposition~\ref{prop:anomalous}) shows that
$\sin\thz$ is the $\mathcal{PT}$-breaking order parameter along
the $\thz$-axis, while fermion number is exactly conserved along
the $\alpha$-axis (Corollary~\ref{cor:anomaly_alpha}).  The full
two-dimensional parameter space $(\thz,\alpha)$ therefore contains
a $\mathcal{PT}$-breaking boundary at $\thz>0$ and a
unitarity-preserving axis at $\thz=0$.  A natural open question
is whether there is a $\mathcal{PT}$ phase transition — an
exceptional point or spectral transition in the sense of
Bender--Boettcher — at some critical curve in the
$(\thz,\alpha)$ plane, and whether this transition is visible in
the conserved charge structure or the soliton spectrum.  The
pseudo-Hermitian metric $\eta = \ee^{\vartheta\gamma^5}$ that
guarantees real spectra for constant complex mass phases may
generalise to a field-dependent metric for the full interacting
system, and the existence of such a metric would be a strong
structural result connecting the integrability and the
non-Hermitian quantum mechanics of the family.

\textit{Quantum integrability.}  The two endpoint systems are both
quantum integrable: the Dirac--sine-Gordon system is equivalent to
the massive Thirring model, whose exact $S$-matrix and quantum
inverse scattering formulation are well
established~\cite{BergKarowski1979,ZamolodchikovZamolodchikov1979,
IzerginKorepin1981}, while the quantum integrability of the
Dirac--sinh-Gordon system follows from the affine Toda structure.
It is natural to ask whether the intermediate system is also quantum
integrable, and whether the $\thz$-deformation is related to a
$q$-deformed $S$-matrix with $q = \ee^{\im\thz}$ in the sense of
quantum groups~\cite{Drinfeld1987,Jimbo1985}.

\textit{Soliton spectrum.}  The sinh-Gordon equation supports no
conventional real solitons, while the sine-Gordon equation has a
rich spectrum of topological kink solitons and breathers.  The
evolution of this spectrum as $\thz$ varies from $0$ to $\pi/2$ is
not yet understood.  In particular, it is an open question whether
the fermion bound state that exists in the strongly-coupled regime
at $\thz=0$ persists for intermediate values of the decoupling
parameter, or whether there is a critical value of $\thz$ at which
it ceases to exist.

\textit{Full Yang--Baxter dictionary.}  Proposition~\ref{prop:YB}
establishes the identification $\eta = \tan(\thz/2)$ at leading
order in the fermion sector within the Toda reduction.  A complete
treatment would extend this beyond leading order, verify
consistency with the full $\eta$-deformed equations of motion, and
connect the fermion sector to the matter fields studied in the
context of non-abelian Toda theories~\cite{FerreiraMiramontes1997,
HolloMiraSchmid2014,DelducMagroVicedo2014}.

\textit{Bosonization as a real-form transition.}  The
$\alpha$-axis connects the Dirac--sinh-Gordon and
Dirac--sine-Gordon systems while preserving fermion number
and staying within the unitary sector.  The Dirac--sine-Gordon
system is equivalent, via Coleman--Mandelstam bosonization, to
the massive Thirring model.  The $\alpha$-deformation therefore
provides a classical integrable path from the fermionic
(Thirring-type) description to the bosonic (sine-Gordon)
description, realised as the action of the $\alpha$-automorphism
rotating the real form of the Lax algebra from $\slR$ to
$\mathfrak{su}(1,1)$ while preserving unitarity.  This suggests
that the Coleman--Mandelstam bosonization — at least at the
classical integrable level — is the real-form transition
$\slR\to\mathfrak{su}(1,1)$ in the Lax algebra of the system,
with the fermion-to-boson map being the $\alpha$-automorphism
acting on the off-diagonal Lax entries.  Establishing this
identification precisely — in particular showing that the quantum
equivalence of the two descriptions corresponds to the quantum
automorphism at the level of the $S$-matrix or the Bethe ansatz
— would be a substantial result connecting the representation
theory of $\slC$ to the non-perturbative structure of $1+1$
dimensional quantum field theory.  For higher-rank algebras
$\mathfrak{g}_\CC$, the analogous real-form transitions would
predict new bosonization dualities connecting matter-coupled
Toda field theories to their dual descriptions, organised by
the Satake diagram of $\mathfrak{g}_\CC$.

\textit{Extension to $\widehat{\mathfrak{sl}}(n)$ and
classification by Satake diagrams.}
The automorphism principle of Remark~\ref{rmk:automorphism}
makes the $\widehat{\mathfrak{sl}}(n)$ extension both natural
and systematic.  For the affine Toda hierarchy associated with
$\widehat{\mathfrak{sl}}(n)^{(1)}$, the inner automorphism
group $\mathrm{Inn}(\mathfrak{sl}(n,\CC))$ has rank $n-1$,
generated by the adjoint actions of the $n-1$ Cartan generators.
Each Cartan generator contributes one independent phase parameter
— a $\thz$-type or $\alpha$-type deformation — and the full
integrable parameter space is a torus $[0,\pi/2]^{2(n-1)}$
whose corners are the distinct real forms of $\mathfrak{sl}(n,\CC)$.
The real forms are classified by the Satake diagram, which for
$\mathfrak{sl}(n,\CC)$ gives the algebras
$\mathfrak{sl}(n,\RR)$, $\mathfrak{su}(p,n-p)$ for
$p=1,\ldots,\lfloor n/2\rfloor$, and $\mathfrak{sl}(n/2,\mathbb{H})$
(for $n$ even).  Each pair of adjacent real forms in this
classification is connected by an integrable path within the
automorphism orbit, and the $\mathcal{PT}$-breaking or
unitarity-preserving character of each path is determined by
whether it deforms the Cartan involution or not.  The
constraints imposed by the zero-curvature condition on which
parameter combinations are accessible — and whether the full
torus $[0,\pi/2]^{2(n-1)}$ is integrable or only a proper
sub-family — is a concrete classification problem whose answer
would give a complete picture of integrable deformations in the
affine Toda setting.

\textit{Superalgebra extension.}  The sinh-Gordon equation admits
a supersymmetric extension associated with the affine Lie
superalgebra $\widehat{\mathfrak{sl}}(2|1)$.  Whether the
$\thz$-deformation extends compatibly to this super-Toda context,
and how the fermionic degrees of freedom of the superalgebra
interact with the Dirac fermion of the present model, are open
questions that touch on the broader programme of classifying
integrable supersymmetric field theories in $1+1$ dimensions.

\ack{The author thanks the Department of Physics at Colorado School of
Mines for support.}

\end{document}